\documentclass[letterpaper]{article} %
\usepackage{aaai2026}
\usepackage{times}  %
\usepackage{helvet}  %
\usepackage{courier}  %
\usepackage[hyphens]{url}  %
\usepackage{graphicx} %
\usepackage{natbib}  %
\usepackage{caption} %
\usepackage{algorithm}
\usepackage{algorithmic}
\usepackage{multirow}
\usepackage{comment}

\usepackage{xcolor}
\usepackage{newfloat}
\usepackage{listings}
\usepackage{amssymb,amsmath,amsthm}
\usepackage{array}
\usepackage{cleveref}

\DeclareCaptionStyle{ruled}{labelfont=normalfont,labelsep=colon,strut=off} %
\floatstyle{ruled}
\newfloat{listing}{tb}{lst}{}
\floatname{listing}{Listing}

\usepackage{thmtools, thm-restate}

\newtheorem{theorem}{Theorem}
\newtheorem{lemma}{Lemma}
\newtheorem{definition}{Definition}

\newtheorem{observation}{Observation}
\newtheorem{example}{Example}

\newcommand{\ev}{\mathbb{E}}
\newcommand{\sw}{\mathrm{UW}}
\newcommand{\ch}{\mathrm{CH}}
\newcommand{\cone}{\mathrm{Cone}}
\newcommand{\uprof}{\vec u} %
\newcommand{\prof}{\vec \sigma} %
\newcommand{\induces}{\triangleright}

\newcommand{\dist}{\mathrm{D}}
\newcommand{\poly}{\mathrm{poly}}
\newcommand{\eps}{\varepsilon}
\newcommand{\elltwo}{\ell^2}
\newcommand{\ellone}{\ell^1}

\newcommand{\uproj}{f_\mathrm{UProj}}
\newcommand{\cent}{\uproj}
\newcommand{\uprojtwo}{f_\mathrm{UProj2}}
\newcommand{\plur}{f_\mathrm{Plur}}
\newcommand{\lslr}{f_\mathrm{LSLR}}
\newcommand{\lslrtwo}{f_\mathrm{LSLR2}}
\newcommand{\pslr}{f_\mathrm{PSLR}}
\newcommand{\mcp}{f_\mathrm{MCP}}

\newcommand{\hr}{f_\mathrm{HR}}
\newcommand{\rd}{f_\mathrm{RD}}

\newcommand{\bigo}{\mathcal{O}}
\newcommand{\kl}[2]{\mathrm{KL}\left(#1 \| #2\right)}
\newcommand{\cW}{\mathcal{W}}

\newcommand{\defeq}{:=}
\newcommand{\abs}[1]{\left\lvert #1 \right\rvert}

\allowdisplaybreaks

\title{Optimized Distortion in Linear Social Choice}

\author{
    Luise Ge,
    Gregory Kehne,
    Yevgeniy Vorobeychik
}
\affiliations{
      Washington University in St. Louis
}

\begin{document}

\maketitle

\begin{abstract}
Social choice theory offers a wealth of approaches for selecting a candidate on behalf of voters based on their reported preference rankings over options.
When voters have underlying utilities for these options, however, using preference rankings may lead to suboptimal outcomes vis-\`a-vis utilitarian social welfare.
Distortion is a measure of this suboptimality, and provides a worst-case approach for developing and analyzing
voting rules when utilities have minimal structure.
However in many settings, such as common paradigms for value alignment, alternatives admit a vector representation, and it is natural to suppose that utilities are parametric functions thereof.

We undertake the first study of distortion for linear utility functions.
Specifically, we investigate the distortion of linear social choice for deterministic and
randomized voting rules.
We obtain bounds that depend only on the dimension of the candidate embedding, and are independent of the numbers of candidates or voters.
Additionally, we introduce poly-time instance-optimal algorithms for minimizing distortion given a collection of candidates and votes.
We empirically evaluate these in 
two real-world domains: recommendation systems using collaborative filtering embeddings, and opinion surveys utilizing language model embeddings, benchmarking several standard rules against our instance-optimal algorithms.
\end{abstract}

\section{Introduction}

Conventional voting rules map a profile of voter preference rankings over a set of candidates (options, outcomes, choices) to a winner.
If, instead of rankings, we were instead to associate candidates with utility values for each voter, in many settings it would be natural to choose a social-welfare-maximizing candidate.
But when voters' preference rankings are generated by latent and unknown utilities, ranking-based voting rules can select candidates that are not welfare-maximizing.
Thus, the limited information about utilities obtained from rankings can be viewed as an imperfect means of approximating a welfare-maximizing choice, and the quality of this approximation---the worst-case ratio of optimal welfare to welfare obtained by a voting rule---is known as \emph{distortion}.
Since its introduction by \citet{procaccia06distortion}, an extensive literature has 
used distortion as means to compare voting rules, identify new rules, and understand the inherent difficulty in different settings of identifying desirable outcomes from ordinal and otherwise incomplete information \cite{anshelevich2021distortion}

If we make no assumptions, no deterministic rules have bounded distortion, and (when there are many voters) the best randomized rule chooses a winning candidate uniformly at random.
Starting with \citet{procaccia06distortion}, a prevalent approach is to assume each voter's utilities are nonnegative and sum to 1 (also see \citet{aziz2020justifications} for discussion of this assumption).
Another mild but incomparable choice is to assume each voter's utilities span the range $[0,1]$, which generalizes approval preferences when each voter has at least one approval and one disapproval.
In all such settings, even optimal voting rules have distortion polynomial in the number of candidates \cite{ebadian2024optimized}.

What if we wish to make an approximately welfare-maximizing choice from among many candidates, and voters' utilities are more structured than unit-sum or unit-range?
Our motivating setting is reinforcement learning from human feedback (RLHF), where options are naturally represented as vector embeddings and voting data consists of rankings. In this context, a voter's utility for an option reflects how well the option aligns with what the voter cares about—mathematically captured as the inner product between the voter's preference vector and the option's embedding. This linear utility model is both a common assumption in RLHF~\citep{zhu2023principled,ge2024axioms} and arises naturally in recommendation systems~\cite{pennock2000social} and multi-objective decision making~\cite{ngatchou2005pareto}, where candidates have feature representations that voters weight according to their preferences. More broadly, structured representation spaces for both voters and candidates are increasingly relevant as AI clones and artificial agents acting on behalf of humans are studied and deployed~\cite{NYT:ModernLove:DatingAppsSuck,Lin:GrindrAI,liang25ai}, motivating the study of parametric utility structures over such spaces. Linear utilities provide a natural starting point for this investigation.

Unlike in many traditional applications of social choice, such as political elections, in RLHF and similar settings the space of options is extremely large.
For example, the space of all possible conversational responses to prompts is so vast as to defy enumeration.
Consequently, distortion bounds that depend on the number of candidates $m$ in these settings have little bite, since $m$ can be exponential---or larger---in the number of features $d$.

We investigate distortion when utility functions are linear, and both the candidates' and voters' embeddings into common feature space $\mathbb{R}^d$ are non-negative and normalized to $1$.
We focus on $\ellone$ normalization, and defer discussion of $\elltwo$ normalization to Appendix~\ref{app:ell-two-normalization}.
Notably, under $\ellone$ normalization, our model recovers the unit-sum setting when $m=d$ and candidates embed to the standard basis of $\mathbb{R}^d$.

A significant benefit of measuring distortion over linear utilities in RLHF and alignment settings is its robustness to the introduction of duplicate candidates, i.e. clones.
This is a criterion of particular importance in the application of social choice methods to alignment problems \cite{procaccia2025clone,berker25independence}.
By way of illustration, consider a profile $\prof$ of voters' rankings of $C$, and construct $\prof_{\cup c'}$ by adding a single duplicate $c'$ of some $c\in C$ which all voters rank (weakly) directly below $c$.
First, can the space of profile-consistent voter utilities over the original candidates $C$ change from $\prof$ to $\prof_{\cup c'}$?
For unit-sum and unit-range utilities, yes; for $d$-dimensional linear utilities, no.
Second, $m$-dependent distortion bounds degrade from $m$ to $m' = m + 1$, while $d$-dependent bounds do not.

\paragraph{Summary of contributions.} 
On the positive side, we introduce three novel voting rules with distortion bound a function of \emph{only the dimension} $d$.
We summarize these results in \Cref{tab:distortion-UBs-ell1}.
The first is the deterministic \emph{max coordinate plurality} (MCP) rule.
Given the candidate embeddings in $\mathbb{R}^d_{\geq 0}$, it first chooses one candidate that maximizes the value of each coordinate, and then selects a plurality winner among this subset.
We show that MCP attains distortion $O(d^3)$ 
(\Cref{thm:MCP-dist-ell1}).
The second and third are randomized rules that are adaptations of the stable lottery rule of \citet{ebadian2024optimized} to both the settings where candidate embeddings into $\mathbb{R}^d$ are known and unknown.
We consider the case when rules can access candidate embeddings to be the standard setting.
Here we construct a \emph{linear stable lotteries} rule, which attains the asymptotically optimal $\Theta(\sqrt{d})$ distortion for linear utilities (\Cref{thm:stable-lottery-ell1}).
When, as in more traditional social choice settings, access to candidate vectors is unavailable, a version of the stable lotteries rule which only samples candidates from suitably-sized (and randomly chosen) committees attains distortion $O(d)$; we dub this the \emph{pure stable lotteries} rule (\Cref{thm:pure-stable-lottery-ell1}).
We also show that random dictatorship attains $O(d^3)$ distortion in this setting (\Cref{thm:RD-distortion-ell1}).

Alongside rules with provable worst-case guarantees, we develop a linear programming (LP)-based approach for computing instance-optimal distortion-minimizing candidates for a given set of options and voter preferences.
While the resulting LP has an infinite set of constraints, we identify an efficient separation oracle that enables efficient optimization over it.
We then empirically evaluate the extent to which our instance-optimal approach outperforms other rules.

\begin{table}[t]
    \centering
    \begin{tabular}{c c c}
        $C$ Embedding & Randomized & Deterministic \\
        \hline
        \hline
        \multirow{2}{*}{Known} & $O(\sqrt{d})$ \:\: (Thm.~\ref{thm:stable-lottery-ell1}) & $O(d^3)$ \:\: (Thm.~\ref{thm:MCP-dist-ell1}) \\
        & $\Omega(\sqrt{d})$ \:\: (Thm.~\ref{thm:random-LB-ell1}) & $\Omega(d^2)$\:\:(Thm.~\ref{thm:det-LB-ell1}) \\
        \hline
        Unknown & $O(d)$ \:\:  (Thm.~\ref{thm:pure-stable-lottery-ell1})  & - \\
    \end{tabular}
    \caption{Distortion in $d$-dimensional linear social choice when candidates and voters are $\ellone$-normalized. Note the unknown setting inherits lower bounds from the known setting.}
    \label{tab:distortion-UBs-ell1}
\end{table}

\paragraph{Further related work.}

Our model diverges from two major strands of prior work in computational social choice: (i) classical frameworks with minimal utility assumptions, and (ii) metric distortion approaches that adopt latent representations, but differ fundamentally in objective and structure.

In classical frameworks, various assumptions are imposed directly on the utility values assigned by voters to candidates—such as unit-sum (each voter's utilities sum to one), unit-range (utilities span a fixed range, e.g., \([0, 1]\)), and approval (where each utility is either 0 or 1). Our approach strictly generalizes this first model by imposing an \(\ell^1\) normalization constraint on voter vectors and treating candidates as standard basis vectors, i.e., with \(m = d\) (\(\ell^\infty\) normalization would similarly generalize the second and third).

Metric distortion approaches are more closely aligned with our setting in that they also assume latent (metric) structure. However, their motivation typically centers on distance-based objectives in utility space---such as minimizing disutility in facility location problems---rather than maximizing social welfare. Although the relation $\mathrm{disutil}_v(c) := 1-u_v(c) = 1-v^Tc$ satisfies the triangle inequality and thus fits into the metric distortion framework, minimizing social cost in this sense is not equivalent to maximizing social welfare—our primary objective. As a result, our approach departs from metric distortion in both its goal and its fundamental structure.

Finally, we remark that while our work builds on the linear social choice framework introduced by \citet{ge2024axioms}, we impose additional structure on the voter and candidate vectors for analysis purpose. 
In addition, although our work and that of \citet{golz2025distortion} share an interest in alignment, their model of the \textit{distortion of alignment} is substantively different from ours. First, they bound utilities within \([0,1]\) but otherwise make no assumptions about latent preference structure. Second, they assume distributions over both voters and candidates, with pairwise comparisons drawn i.i.d. from these distributions; distortion is then evaluated in expectation over this distribution. In contrast, we make no distributional assumptions and evaluate distortion in the worst case over all possible voter-candidate profiles. Finally, we consider truthful ordinal rankings, while \citet{golz2025distortion} assume a Bradley--Terry noise model for reported preferences, which enables some access to voter preference intensities. This reflects different modeling goals: we aim for robustness under accurate ordinal preference reporting, while they seek insights under probabilistic reporting with minimal utility structure assumptions.

\section{Model and Preliminaries}
\label{sec:model}

We consider a setting with $n$ voters $V$ and $m$ candidates $C$, where each voter $v \in V$ and each candidate $c \in C$ corresponds to a vector in $\mathbb{R}^d_{\geq 0}$. The $d$ dimensions correspond to positive attributes that determine both voters' preferences and candidates' characteristics. We use superscripts to denote coordinates of a vector.

We assume that the utility of a voter $v$ for a candidate $c$ is given by the inner product
\(u_v(c) = v^\top c.\)
This \textit{utility function} $u_v: C \to \mathbb{R}_{\ge 0}$, in turn, induces a \textit{preference ranking} $\sigma_v$ where $c \succ_v c'$ if $u_v(c) \geq u_v(c')$, with ties broken arbitrarily. 
Below we will introduce some constraints for $v$ and $c$, and use $\mathcal{U}$ to denote the set of all feasible utility functions induced by such $v$ and $c$.

Let $\vec{u} = \{u_v\}_{v \in V}$ be the \textit{utility profile} of all voters, and let $\prof = \{\sigma_v\}_{v \in V}$ be the \textit{preference profile} of all voters. We use the notation $\uprof \induces \prof$ to indicate that $u_v$ is consistent with $\sigma_v$ for each voter $v \in V$. 
Voting rules have access to the rankings $\prof$, but \textit{not} to the utilities $\uprof$.

If we know the utility functions $u_v$ of all voters, a natural candidate selection criterion is \emph{utilitarian welfare}, which is the sum of voter utilities:
\(\sw(\uprof, c) \defeq \sum_{v \in V} u_v(c)\). In this case the winning candidate $c^*(\uprof)$ is the one with highest welfare, i.e., $c^*(\uprof) \in \arg\max_{c \in C} \sw(\uprof, c)$.

Let $f$ be a (possibly randomized) voting rule that maps a preference profile $\prof$ to a winner $c \in C$. For a fixed profile $\prof$, the \textit{instance distortion} of $f$ on $\prof$ is the worst-case ratio between the optimal utilitarian welfare and the utilitarian welfare of $f$; that is,
\[
    \dist(f, \prof) \defeq \max_{\uprof \induces \prof} \frac{\sw(\vec{u}, c^*(\vec{u}))}{\mathbb{E}_{c \sim f(\prof)} [\sw(\vec{u}, c)]}.
\]

For theoretical results, we are primarily interested in the overall \textit{distortion} of $f$, which is the worst case over profiles:

\[
    \dist(f) \defeq \max_{\uprof \induces \prof, \uprof \in \mathcal{U}^n }\dist(f, \prof).
\]

We impose the following structural assumptions. 

\begin{itemize}
    \item \textbf{Non-negativity.} All voter and candidate vectors lie in the positive orthant, i.e., $v, c \in \mathbb{R}^d_{\ge 0}$. This ensures all utilities are non-negative, and avoids a mixed-sign objective. 

  \item \textbf{Normalization.} 
  We will assume $\|v\|_p = 1$ and $\|c\|_p = 1$ for all $v, c$.
  For $v \in V$ this can be viewed as constraining all voters to have equal influence on the mechanism (c.f. \citet{aziz2020justifications}), and for both candidates and voters as identifying the \emph{relative} magnitude of embedding components.
  We focus on $\ell^p$ norms for $p=1$ in the main body,\footnote{Note that this together with non-negativity is equivalent to assuming $v,c \in \Delta_d$, where $\Delta_k$ denotes the $k$-dimensional simplex.} but also provide results for $\elltwo$ normalization in the supplemental material.

    \item \textbf{Expressiveness.} For theoretical analysis, we further assume $V \subset \mathrm{Cone}(C)$, meaning every voter can be expressed as a non-negative linear combination of candidates. 
    Intuitively, this can be understood as requiring that the candidate set is rich enough to describe voter preferences.
    This excludes the case where voters have $0$ utility for all alternatives, ensuring each voter has meaningful preferences over the candidate set.
\end{itemize}

Normalization of both $v$ and $c$ is necessary to keep the model well-posed and the distortion finite, as the following example illustrates:
\begin{example} \label{ex:normalization-necessary}
    Consider $n-1$ voters who prefer $c_1$ to $c_2$ and one voter who strongly prefers $c_2$ due to a large utility spike in one coordinate. 
    If candidate $c_2$ has an unbounded entry, it can yield unbounded utilitarian welfare even though $n-1$ voters rank $c_1$ above $c_2$. 
\end{example}

We study several established voting rules. 
We define them here for convenience.

\begin{definition}(Randomized Scoring Rules (RSRs))\label{def:randomized-scoring-rules} 
Let $\vec{s} = (s^1, \ldots, s^m)$ be a scoring vector with $s^1 \geq s^2 \geq \cdots \geq s^m \geq 0$. For a candidate $c \in C$, let $\operatorname{rank}_v(c)$ denote the position of $c$ in voter $v$'s ranking (with $\operatorname{rank}_v(c) = 1$ meaning $a$ is ranked first). Then the score assigned to $c$ by agent $v$ is:
\[
    \operatorname{score}_v(c, \vec{s}) \defeq s^{\operatorname{rank}_v(c)}.
\]
The total score of $c$ across all voters is:
\[
    \operatorname{score}_V(c, \vec{s}) \defeq \sum_{i \in n} \operatorname{score}_v(c, \vec{s}).
\]
The \emph{randomized scoring rule} $f^{\text{rand}}_{\vec{s}}$ selects each alternative $c \in C$ with probability proportional to its total score:
\[
    \Pr[f^{\text{rand}}_{\vec{s}}({\prof}) = c] \defeq \frac{\operatorname{score}_V(c, \vec{s})}{n \cdot \|\vec{s}\|_1}.
\]

\end{definition}

Notable examples include \textit{Random Dictatorship}, which corresponds to the plurality scoring rule $\vec{s} = (1, 0, \ldots, 0)$, and the \textit{Randomized Harmonic} rule \cite{boutilier15optimal}, which uses $\vec{s} = (1 + \frac{H_m}{m}, \frac{1}{2} + + \frac{H_m}{m}, \ldots, \frac{1}{m}++ \frac{H_m}{m})$, where $H_m$ is the $m^{th}$ harmonic number.

We will also make use of stable lotteries, which were shown to exist for any preference profile by \citet[Lemma 4]{cheng2020group}.
\begin{definition}[Stable Lotteries]
\label{def:stable-lotteries}
    Given a preference profile $\prof$, a committee $W$, and a candidate $c$, let $S_c(W) \defeq \{v \in V: c \succ_v W\}$ be the set of voters who prefer $c$ to all alternatives in $W$.
    We say a distribution $\cW$ over committees of size $k$ is a \emph{stable lottery} if for all $c \in C$,
    \[
        \ev_{W \sim \cW} \left[ \abs{S_c(W)} \right] \leq \frac{n}{k}.
    \]
\end{definition}
This has seen much recent use in computational social choice; our direct inspiration is its application by \citet{ebadian2024optimized} to the design of distortion-optimal randomized rules in the unit-sum and related settings.

\paragraph{Special case: unit-sum utilities.}
The special case in which each candidate is a standard basis vector, i.e. $C = \{e_1, \ldots, e_d\}$, recovers the well-studied unit-sum model, 
in which distortion was first introduced \cite{procaccia06distortion}, and for which 
the worst-case distortion-optimal randomized rule is has distortion $\Theta(\sqrt{m})$ \cite{boutilier15optimal,ebadian2024optimized}, while the worst-case optimal deterministic rule has distortion $\Theta(m^2)$~\cite{caragiannis2017subset}.
By generalizing this setting, we inherit its lower bounds 
for both deterministic and randomized rules 
(\Cref{thm:det-LB-ell1,thm:random-LB-ell1}).

\section{Deterministic Rules}

What utility guarantees can deterministic rules provide in this setting? 
How well do prominent deterministic rules fare?
We begin with a useful lower bound on every individual voter's utility for their favorite candidate.
Throughout, we use $\ch(C)$ to denote the \textit{convex hull} of a set $C$.
\begin{lemma}
\label{lem:voter-util-LB}
   For any $v$ and any candidates ${C=\{c_j\}_{j \in [m]}}$, the maximum utility of $v$ is at least
   $ \max_{c \in C} u_v(c) \geq \frac{1}{d}$.
\end{lemma}

\begin{proof}
    Let $\tilde{c} \defeq \arg\max_{c \in C} v^T c$.
    For any $\alpha \in \Delta_m$,
    $v^T\tilde{c} \ge v^T (\sum_{i \in [m]} \alpha_i c_i) = \sum_{i \in [m]} \alpha_i(v^Tc_i)$ since the maximum upper bounds any convex combination. Moreover, for $\ellone$ normalization, $v \in \cone(C)$ implies $v \in \ch(C)$, i.e. $v=\sum_{i\in [m]} \beta_i c_i $ for some $\beta \in \Delta_m$. Substituting $\alpha$ with this specific $\beta$, we have
    $v^T\tilde{c} \ge v^Tv \ge \frac{1}{d}$ by Cauchy-Schwarz (in particular, since $\|v\|_2 \cdot d \ge \|v\|_1=1$).
\end{proof}
An analogous claim (Lemma~\ref{lem:voter-util-LB-ell2} in Appendix \ref{app:ell-two-normalization}) holds for $\elltwo$ normalization.

A natural deterministic rule is plurality ($\plur$). 
In the unit-sum setting, $\plur$ is known to have distortion $\Theta(\min(n,m) \cdot m)$, which is asymptotically optimal \cite{caragiannis2017subset}. 
We show that $\plur$ attains worst-case distortion $\Theta(\min(n,m)\cdot d)$ for linear utilities. 
This matches the unit-sum guarantee when $d=m$, but scales poorly for $m,n \gg d$.

\begin{theorem}\label{thm:plurality-distortion-ell1}
$\dist(\plur)=\Theta(\min(m,n) \cdot d).$
\end{theorem}
\begin{proof}
    We begin with the upper bound.
    If $m \le n$, by Lemma~\ref{lem:voter-util-LB}, we know the plurality winner must have social welfare at least $\frac{n}{m}\cdot\frac{1}{d}$.  Otherwise, the plurality winner still must have social welfare at least $\frac{1}{d}$. And the best possible candidate's social welfare is at most $n$.
    Hence, $\dist(\plur)=\bigo(\min(m,n)\cdot d)$.
   
    We now demonstrate an instance for which this is tight.
    Let $\mu \defeq (\frac{1}{d}, \ldots, \frac{1}{d})$ and consider a candidate set where $c_1 = \mu$ and $c_2 = \cdots = c_m = (1, 0, \ldots, 0)$. First, assume $n \geq m$. Construct a profile where two voters vote for $c_1$ and each of the remaining $n - 2$ voters votes for a distinct candidate among $c_2, \ldots, c_m$. Let $v_1 = v_2 = \mu$ and $v_3, \ldots, v_n = (1, 0, \ldots, 0)$. The social welfare of $c_1$ is $2/d$, while any other candidate has welfare $1$ from each of the other $n - 2$ voters and $1/d$ from the first two, giving at least $n - 2 + 2/d$. Thus, the distortion is
    \(
        \dist(\plur, \prof) = \frac{n - 2 + 2/d}{2/d} \in \Omega(n\cdot d).
    \)
    
    Now assume $n < m$ and suppose $m$ is divisible by $n$. Let exactly $\frac{n}{m} + 1$ voters rank $c_1$ first (denote this group as $\mathcal{G}_1$), and assign the remaining voters evenly among the other candidates. Let $v_j = \mu$ for $j \in \mathcal{G}_1$, and $v_j = (1, 0, \ldots, 0)$ for all other voters. Then the social welfare of $c_1$ is at most $(n/m + 1)/d$, while the best alternative has social welfare at least $n - n/m - 1 + (n/m + 1)/d$. The distortion becomes
    \[
        \dist(\plur, \prof) = \frac{n - \frac{n}{m} - 1 + \frac{1}{d}(\frac{n}{m} + 1)}{\frac{1}{d} \cdot (\frac{n}{m} + 1)} = \Omega(m\cdot d).
    \]
    In both cases, the lower bound is $\Omega(\min(m,n) \cdot d)$.
\end{proof}

Is this dependence on $m$ and $n$ unavoidable for all deterministic rules, or can it be circumvented?
For deterministic rules, we inherit lower bounds from the work of \citet{caragiannis2017subset} in the unit-sum setting.

\begin{theorem}\label{thm:det-LB-ell1}[Theorem 1 of \citet{caragiannis2017subset}]
Suppose $m, n \geq d$. 
Then for any deterministic voting rule $f$, we have $\dist(f) = \Omega(d^2)$.
\end{theorem}

This lower bound depends only on $d$, and raises the question of whether it is possible to improve upon plurality by avoiding a dependence on $n$ and $m$ in the distortion bound. 
Indeed it is: we introduce a rule we dub \emph{maximum coordinate plurality (MCP)} and denote by $\mcp$.
\begin{definition}[Maximum Coordinate Plurality]
\label{def:mcp}
    Let
    \[
        \hat{C} \defeq \left\{c_i \in C \mid c_i = \arg\max_{c \in C} c^i \text{ for each } i \in [d] \right\}
    \]
    be a set of at most $d$ candidates such that for each coordinate $i$, a candidate maximal in that coordinate is included. 
    The \emph{maximum coordinate plurality rule} $\mcp$ then restricts the profile to $\hat{C}$ and selects the plurality winner.
\end{definition}

\begin{theorem}\label{thm:MCP-dist-ell1}
   $\dist(\mcp)=O(d^3)$.
\end{theorem}
\begin{proof}
    For any fixed voter $v$, it has its maximum coordinate at least $\frac{1}{d}$; call this coordinate $i$.
    Since $v \in \ch(C)$, the candidate $c_i \in \hat C$ therefore has $i$-th coordinate at least $\frac{1}{d}$, and so the welfare conferred to $v$ by their favorite choice in $\hat C$ is at least $\max_{c \in \hat C} v^T c \geq \frac{1}{d^2}$.
    Since $|\hat{C}| \leq d$, this plurality winner $\hat{c} \in \hat C$ receives at least $\frac{n}{d}$ votes, and so
    \(
        \sw(\hat{c}) \geq \frac{n}{d} \frac{1}{d^2}.
    \)
    As $\max_{c \in C} \sw(c) \leq n$, the claim follows.
\end{proof}

Though this guarantee still exceeds the lower bound by a factor of $d$, the $\Omega(d^2)$ lower bound is already quite large. 
Can this be improved by randomizing over candidates?

\section{Randomized Rules}

As previously mentioned, we also inherit a distortion lower bound for \emph{any} randomized rule from the unit-sum setting.

\begin{restatable}{theorem}{randomlb} \label{thm:random-LB-ell1}\cite{boutilier15optimal}
    Suppose $m \geq d$ and $n \geq \sqrt{d}$. 
    Then there exists a preference profile $\prof$ such that for any randomized rule $f$, we have $\dist(f,\prof) = \Omega(\sqrt{d})$.
\end{restatable}

We include a proof for completeness in Appendix \ref{app:omitted-proofs}.

\subsection{Known Candidate Embeddings}
\label{sec:randomized-known-embeddings}
Consider the center of the simplex $\Delta_d$, denoted by $\mu \defeq (\frac{1}{d}, \ldots, \frac{1}{d})$. 
Observe that $\sw(\mu) = \frac{n}{d}$, which implies a distortion at most $d$ is possible when $\mu \in C$. 
However in general, $\mu \not\in \ch(C)$.
This motivates the goal of \textit{approximating} the uniform candidate 
$\mu$ by a point within the convex hull of candidates $\ch(C)$, assuming the candidate embeddings are known.
We therefore introduce a randomized rule whose output distribution matches the \textit{reverse information projection} of $\mu$ onto $\ch(C)$.

\begin{definition}[Uniform Projection Rule (\(\uproj\))]
\label{def:uproj}
    Given candidate locations \( C \subset \mathbb{R}^d_{\ge 0} \), the \emph{uniform projection rule} \( \uproj \) defines a distribution \( \{p_c\}_{c \in C} \) over candidates such that the expected candidate vector
    $\hat{c} \defeq \sum_{c \in C} p_c \cdot c$
    minimizes the Kullback–Leibler (KL) divergence from the uniform vector, i.e.,
    \(
        \hat{c} \defeq \arg\min_{x \in \ch(C)} \kl{\mu}{x}.
    \)
  
\end{definition}

\begin{theorem}\label{thm:middle-dist-ell1}
     The expected welfare of $\uproj$ is at least $\sw(\uproj) \geq \frac{n}{d}$.
     As a consequence, $\dist(\uproj)=O(d)$.
\end{theorem}

\begin{proof}[Proof of \Cref{thm:middle-dist-ell1}]
    First observe that
    \[
        \kl{\mu}{x}
        =\sum_{i=1}^d\mu^i\ln\frac{\mu^i}{x^i}
        ={\sum_i\mu^i\ln\mu^i}
        \;-\;\sum_{i=1}^d\mu^i\ln x^i,
    \]
    so minimizing $\kl{\mu}{x}$ over the convex set $\ch(C)$ is
    equivalent to minimizing the smooth, convex function
    \(
        f(x)=-\sum_{i=1}^d\mu^i\ln x^i
    \)
    subject to $x\in \ch(C)$.  
    The first‐order optimality condition gives 
    \(
        \nabla f(x^*)^\top(v-x^*) \;\ge\;0
    \)
    for every voter $v\in \ch(C)$.
    Since
    \(
        \frac{\partial f}{\partial x^i}(x)
        =-\frac{\mu^i}{x^i},
    \)
    we have
    \[
        \nabla f(x^*)^\top(v-x^*)
        =-\sum_{i=1}^d\frac{\mu^i}{x^{*i}}(v^i - x^{*i})
        \;\ge\;0 
    \]
    \[
        \Longrightarrow\;\;
        \sum_{i=1}^d\frac{\mu^i\,v^i}{x^{*i}}\;\le\;\sum_{i=1}^d\mu^i
        =1.
    \]
    Since $\mu^i=\frac{1}{d}$, this yields
    \(
        \sum_{i=1}^d\frac{v^i}{x^{*i}}\;\le\;d.
    \)
    Now define the auxiliary sequences
    $a_i \defeq \sqrt{\frac{v^i}{x^{*i}}}$ and $b_i \defeq \sqrt{v^i\,x^{*i}}$.
    Then
    \(
        \sum_{i=1}^d a_i\,b_i
        =\sum_{i=1}^d v^i
        =1
    \)
    and, by, Cauchy–Schwarz,
    \begin{align*}
        1 = \Bigl(\sum_i a_i b_i\Bigr)^2
        &\le\Bigl(\sum_i a_i^2\Bigr)\Bigl(\sum_i b_i^2\Bigr)\\
        &=\Bigl(\sum_i\tfrac{v^i}{x^{*i}}\Bigr)\Bigl(\sum_i v^i\,x^{*i}\Bigr).
    \end{align*}
    Combining with $\sum_i v^i/x^{*i}\le d$ gives $\sum_{i=1}^d x^{*i}\,v^i \ge 1/d$.
    Now as we have $n$ voters and the above inequality holds for arbitrary $v$, the total utilitarian welfare is at least $\frac{n}{d}$.
    
    The bound on the distortion of $\uproj$ then follows because the maximum utility for any voter is at most $1$.
\end{proof}

This randomized $O(d)$-distortion rule can be seen as generalizing uniform candidate selection in the unit-sum setting. A natural candidate for improving upon this is the harmonic rule $\hr$, which obtains near-optimal distortion $\Theta(\sqrt{m \log m})$~\cite{boutilier15optimal,bhaskar2018truthful} for unit-sum utilities. 
However this performance does not generalize to linear utilities, wherein the introduction of duplicate candidates does not constrain the underlying utility profile.
Instead, it turns out that here $\hr$ has distortion unbounded in $d$; we discuss this and other randomized positional scoring rules in \Cref{sec:unknown-candidates}. 

Fortunately, \Cref{thm:middle-dist-ell1} does indirectly lead to distortion sublinear in $d$.
Using it, we may adapt the stable lottery rule of \cite{ebadian2024optimized} to the linear utilities setting to achieve a better distortion bound.

\begin{definition}[Linear Stable Lottery Rule ($\lslr$)]
\label{def:linear-slr}
    Given a stable lottery $\cW$ over committees of size $k=\sqrt{d}$, the \emph{linear stable lottery rule} $\lslr$ on profile $\prof$ chooses each $c \in C$ with probability $\frac{1}{2\sqrt d}\Pr_{W \sim \cW(\prof)}[c \in W] + \frac{1}{2} \Pr_{c' \sim \cent(\prof)}[c' = c]$.
\end{definition}

Here $\uproj$ takes the place of uniform selection.
\begin{theorem} 
\label{thm:stable-lottery-ell1} $\dist(\lslr)=O(\sqrt{d}).$
\end{theorem}
\begin{proof} %
    We follow the proof of \citet{ebadian2024optimized} closely, albeit in our notation. 
    The principal difference is our use of $\uproj$ instead of uniform selection over candidates.

    To that end, we combine \Cref{def:stable-lotteries} and \Cref{thm:middle-dist-ell1} in order to relate the expected welfare conferred from each part of $\lslr$.
    For any committee $W$ of size $k$ and any $a \in W$, let $S_a(W) \subseteq V$ denote the voters for which $a \succ_v W$, and $\overline S_a(W)$ its complement. 
    Then
    \begin{align}
        \sum_{v \in V} \sum_{c \in W} u_v(c) 
        &= \sum_{v\in S_a(W)} \sum_{c \in W} u_v(c)  + \sum_{v \in \overline S_a(W)} \sum_{c \in W} u_v(c) \notag \\
        &\geq  \sum_{v\in S_a(W)} \sum_{c \in W} u_v(c)  + \sum_{v \in \overline S_a(W)} u_v(c^*) \notag \\
        &\geq  \sum_{v\in S_a(W)} (u_v(c^*) - 1)  + \sum_{v \in \overline S_a(W)} u_v(c^*) \notag \\
        &= \sum_{v\in V} u_v(c^*)  - |\{S_a(W)\}|.
	\end{align}
	Taking the expectation over a stable lottery $\cW$ and applying \Cref{def:stable-lotteries}, we obtain
	\begin{align}
		\ev_{W \sim \cW} \left[  \sum_{c \in W} \sw(c) \right]
        &\geq \sw(c^*)  - \ev_{W \sim \cW} \left[ | \{S_a(W)\} | \right] \notag \\
        &\geq \sw(c^*)  - \frac{n}{k} \notag \\
        &\geq \sw(c^*)  - \frac{d}{k} \ev_{c \sim \cent }\left[ \sw(c) \right], \label{eq:stab-lot-vs-unif-sample}
	\end{align}
	where the last step follows from \Cref{thm:middle-dist-ell1}.

    We now let $x = \frac{1}{2} x_1 + \frac{1}{2} x_2$ denote the distribution over $C$ of $\lslr$, where $x_1$ is the stable lottery part and $x_2$ is the $\cent$ part.
    Then applying \eqref{eq:stab-lot-vs-unif-sample} and letting $k = \sqrt{d}$, 
    \begin{align}
        k \cdot \sw(x_1) &= \ev_{W \sim \cW} \left[  \sum_{c \in W} \sw(c) \right] \notag \\
        &\geq \sw(c^*)  - \frac{d}{k} \ev_{c \sim \cent }\left[ \sw(c) \right] \notag \\
        k \cdot \sw(x_1) + \frac{d}{k} \sw(x_2)  &\geq \sw(c^*)  \notag \\
        \ev_{c \sim \lslr } \left[ \sw(c) \right]  &\geq \frac{1}{2\sqrt{d}} \sw(c^*).  \notag \\
        \frac{\sw(c^*)}{\ev_{c \sim \lslr } \left[ \sw(c) \right] } &\leq 2\sqrt{d} \notag
    \end{align}
    This proves the claim.
\end{proof}

Computing a stable lottery $\cW$ requires only ordinal information; however in order to identify the distribution from which $\cent$ samples, the embedding of $C$ into $\mathbb{R}^d_{\geq 0}$ must be known by the rule.
This raises the question of what can be done with ordinal preferences when both voter \emph{and} candidate embeddings are unknown.

\subsection{Unknown Candidate Embeddings}
\label{sec:unknown-candidates}

Even when the locations of the candidate vectors $c \in \mathbb{R}^d_{\ge 0}$ are unknown, some established rules exhibit distortion that is bounded as a function of $d$.
Our primary example is random dictatorship, for which proof is deferred to Appendix~\ref{app:omitted-proofs}.

\begin{restatable}{theorem}{randomdictator}
    \label{thm:RD-distortion-ell1}
    $\dist(\rd) = \Omega(d)$, and also $\dist(\rd) = O(d^3)$.
\end{restatable}

One feature of this model is that the worst-case distortion of many randomized positional scoring rules (Appendix \ref{def:randomized-scoring-rules}) is quite similar. 
In particular, consider the RSRs where, for given $m$, the score vector is given by $\vec{s} = \frac{1}{S_m}(s_1, s_2, \ldots, s_m)$ for some fixed sequence $s_1, s_2, \ldots$, and where $S_i = \sum_{j \leq i} s_j$.
(Note that this contains $\rd$ and $\hr$ and the uniform distribution over candidates, but not Borda.)
Then there are two cases. 
If $S_m \rightarrow S$ converges, then the distortion of $f_s$ on \emph{any} instance is within $S/s_1$ of $\rd$, and we can make the behavior of $f_s$ approach that of $\rd$ by cloning each candidate sufficiently many times. 
And if $S_m$ diverges, then cloning bad candidates can lead to poor performance on instances which are easy for $\rd$. 
We illustrate this by relating the performance of $\hr$ to that of $\rd$, the proof of which also appears in Appendix \ref{app:omitted-proofs}.

\begin{restatable}{theorem}{harmonic}
    \label{thm:harmonic-vs-RD}
    The distortion of $\hr$ is unbounded as a function of $d$.
    However, for any profile $\prof$, 
    \[
        \dist(\hr, \prof) \leq \dist(\rd, \prof) \cdot ( \log m + 1).
    \]
\end{restatable}

Thus, the distortion of $\hr$ is never much better than that of $\rd$, and if---as we expect---the worst-case profiles for $\rd$ have $\poly(d)$ candidates, then its worst-case distortion cannot be much better than that of $\rd$, even for small $m$.

Can we improve upon this bound when candidate embeddings are unknown?
The role that $\cent$ plays in the design and analysis of $\lslr$---and the role uniform selection plays in the stable lotteries for \citet{ebadian2024optimized}---is in some sense both an absolute lower bound on welfare when the maximum candidate welfare is not large, and sample access to it.
We might then let $\rd$ take the role of $\cent$ in $\lslr$ to attain distortion $O(d^{3/2})$ for unknown embeddings.

However even in the absence of a better rule, it turns out stable lotteries can still be used by furnishing a direct lower bound on the maximum welfare.
Consider the following:

\begin{definition}[Pure Stable Lottery Rule ($\pslr$)]
	\label{def:pure-slr}
	Given a stable lottery $\cW$ over committees of size $2d$, the \emph{pure stable lottery rule} $\pslr$ chooses $c \in C$ w.p. $\frac{1}{2 d}\Pr_{W \sim \cW}[c \in W]$.
\end{definition}

Using larger committees and that $\max_{c \in C} \sw(c) \geq \frac{n}{d}$, we avoid the need for a second sampling component.

\begin{restatable}{theorem}{purestablelottery}
    \label{thm:pure-stable-lottery-ell1}
    $\dist(\pslr)=O(d).$
\end{restatable}

\begin{proof}[Proof (sketch)]
	The difference between this and the proof of Theorem \ref{thm:stable-lottery-ell1} is that the lottery is over committees of size $k = 2d$, and that we do not use $\cent$ as in \eqref{eq:stab-lot-vs-unif-sample}.
	We instead look to \Cref{thm:middle-dist-ell1}, which demonstrates there always exists a distribution $\{p_c\}$ over candidates such that $\ev_{c \sim \cent}\left[\sw(c)\right] \geq \frac{n}{d}$. 
	Then by averaging, there always exists a $c \in C$ such that $\sw(c) \geq \frac{n}{d}$, and in particular $\sw(c^*) \geq \frac{n}{d}$. 
	Picking things up just before \eqref{eq:stab-lot-vs-unif-sample} with $k = 2d$,
	\begin{align*}
		\ev_{W \sim \cW} \left[  \sum_{c \in W} \sw(c) \right]
		&\geq \sw(c^*)  - \ev_{W \sim \cW} \left[ | \{S_a(W)\} | \right] \notag \\
		&\geq \sw(c^*)  - \frac{n}{2d} \notag 
        \geq \frac{1}{2} \cdot \sw(c^*). \label{eq:stab-lot-alone}
	\end{align*}
	The rest of the proof proceeds as before, though without the need for $x_2$.
	It appears in full in Appendix \ref{app:omitted-proofs}.
\end{proof}

\section{Optimizing Distortion}

In the preceding sections we analyzed the asymptotic behavior of voting rules, but the lower bounds are often driven by pathological cases.
In practice, we are more interested in this welfare approximation guarantee for a given profile.
That is:
\textit{Given a preference profile, can we compute the best possible distortion—and design a rule that achieves it?}

The computational tractability of the distortion-instance-optimal rule in the unit-sum setting was established by \citet{boutilier15optimal} and clarified by \citet{ebadian2024computational}. 
For linear utilities we also answer both questions affirmatively, showing that instance-optimal rules—both deterministic and randomized—can be computed efficiently. %

\begin{definition}[Feasible Region for \(\bar v\)]
\label{def:feas-avg-voter}
    Let 
    \(
    \mathcal F_j \;=\;\bigl\{\,v\in \Delta_d : \langle c_a - c_b, v\rangle \ge 0\;\text{whenever }c_a\succ_j c_b\bigr\}
    \)
    be the set of voter vectors consistent with $\sigma_{v_j}$.  Then the feasible region for the average voter 
    \(\bar v = \tfrac1n\sum_{j=1}^n v_j\)
    is
    \(
    \mathcal F
    \;=\;
    \frac1n\sum_{j=1}^n \mathcal F_j.\)
\end{definition}

\begin{theorem}[Deterministic Instance-Optimal Rule]
\label{thm:approx-instance-opt}
    Given a preference profile \( \prof \) and any \( \varepsilon > 0 \), one can compute a deterministic voting rule \( f \) such that
    \(
    \dist(f, \prof) \le \min_{c \in C} \dist(c, \prof) + \varepsilon
    \)
    in time \( \mathrm{poly}(n, m, \log \frac{1}{\eps}).\)
\end{theorem}

\begin{proof}
Let \( \mathcal{F} \) be the feasible region for the average voter vector \( \bar{v} \) induced by \( \prof \) (\Cref{def:feas-avg-voter}). For each pair \( c_1, c_2 \in C \), we aim to compute the best possible upper bound on
\(
\frac{\sw(c_1)}{\sw(c_2)} = \frac{\bar{v}^\top c_1}{\bar{v}^\top c_2}
\)
over all \( \bar{v} \in \mathcal{F} \). Since we cannot compute this ratio exactly without knowing \( \bar{v} \), we upper-bound it by the largest value \( \frac{1}{\beta} \) such that
\(
\bar{v}^\top c_1 \le \frac{1}{\beta} \bar{v}^\top c_2 \ \forall \ \bar{v} \in \mathcal{F}.
\)

Equivalently, for fixed \( c_1, c_2 \), we find the largest \( \beta_{c_1,c_2} \in [0,1] \) such that the following LP has a non-negative optimal value:
\(
\min_{\bar{v} \in \mathcal{F}} \langle \beta c_1 - c_2, \bar{v} \rangle.
\)
This is doable via binary search over \( \beta\) to precision \( \varepsilon \), requiring \( O(\log \frac{1}{\eps}) \) iterations per pair.

Once the distortion bounds are computed for all \( O(m^2) \) candidate pairs, we select the candidate \( \hat{c} \in C \) minimizing the worst-case distortion
\(
    \hat{c} \defeq \arg\min_{c \in C} \max_{c' \in C}\beta_{c',c}.
\)
Returning this candidate defines the deterministic rule \( f \), and by construction we have
\(
\dist(f, \prof) \le \min_{c \in C} \dist(c, \prof) + \varepsilon.
\)
Since each linear program is solvable in polynomial time in \( n, m, d \), and the number of binary search iterations is \( O(\log \frac{1}{\eps}) \), the total runtime is \( \mathrm{poly}(n, m, d, \log \frac{1}{\eps}) \).
\end{proof}

Note that the procedure described above can be used to compute the distortion of any rule on a given preference profile. We also apply it in our experiments.

\begin{theorem}[Randomized Instance-Optimal Rule]
\label{thm:rand-instance-opt}
Given a preference profile \( \prof \) and any \( \varepsilon > 0 \), a randomized voting rule \( f \) with
\(
\dist(f, \prof) \le \min_{f' \in \Delta(C)} \dist(f', \prof) + \varepsilon
\)
can be computed in time \( \mathrm{poly}(n, m,\frac{1}{\eps}) \), where the minimum is taken over all distributions over candidates.
\end{theorem}

\begin{proof}
Let \( p \in \Delta_C \) be a distribution over candidates, and let \( \hat{c} \defeq \sum_{i=1}^m p_i c_i \) be the expected candidate vector under \( p \). We aim to find a distribution \( p \) that minimizes the distortion:
\[
\dist(p, \prof) = \sup_{\bar{v} \in \mathcal{F}} \frac{\max_{c \in C} \langle c, \bar{v} \rangle}{\langle \hat{c}, \bar{v} \rangle}.
\]
Since this expression may be unbounded, we instead maximize its reciprocal \( \beta \in [0,1] \), subject to the constraint:
\[
\langle \hat{c}, \bar{v} \rangle \ge \beta \cdot \max_{c \in C} \langle c, \bar{v} \rangle \quad \forall\, \bar{v} \in \mathcal{F}.
\]

This yields the following optimization problem:
\[
\begin{aligned}
\max_{\underset{\beta \in [0,1]}{p \in \Delta(C)}} \left\{ \beta : \sum_{i=1}^m p_i \langle c_i, \bar{v} \rangle \ge \beta \cdot \max_{c \in C} \langle c, \bar{v} \rangle \quad\forall\, \bar{v} \in \mathcal{F} \right\}.
\end{aligned}
\]

Although \( \mathcal{F} \) contains infinitely many constraints, we can apply the ellipsoid method using a separation oracle. For a candidate solution \( (\hat{\beta}, \hat{p}) \), feasibility reduces to verifying:
\[
\sum_{i=1}^m \hat{p}_i \langle c_i, \bar{v} \rangle \ge \hat{\beta} \cdot \langle c, \bar{v} \rangle \qquad \forall\, c \in C,\ \bar{v} \in \mathcal{F}.
\]

For each \( c \in C \), this can be checked by solving an LP:
\[
\begin{aligned}
\min_{\bar{v} \in \mathcal{F}} \quad & \left\langle \sum_{i=1}^m \hat{p}_i c_i - \hat{\beta} \cdot c,\; \bar{v} \right\rangle.
\end{aligned}
\]

If the minimum is nonnegative for all \( c \in C \), the solution is feasible; otherwise, the auxiliary LP provides a separating hyperplane. Since each \( \bar{v} \in \mathcal{F} \) is defined via \( O(m^2 n) \) linear constraints, each LP is of polynomial size. Running the ellipsoid method with this oracle to precision \( \varepsilon \) yields a solution in time \( \mathrm{poly}(n, m, d, \frac{1}{\eps}) \).
\end{proof}

\noindent \textbf{Practical implementation.}
While the ellipsoid method offers a polynomial-time theoretical guarantee, it is often impractical in real-world settings. 
We can instead use \textit{column generation}~\citep{dantzig1960decomposition}, which is often more efficient in practice (see Appendix~\ref{app:pseudocode}). 
Notably, the method relies only on pairwise comparisons to construct constraints for each voter, making it also applicable for partial rankings.

\section{Experiments}
\label{sec:experiments}

We evaluate our instance-optimal rules alongside $\uproj$, $\mcp$, $\lslr$, and several other common rules, on two real-world datasets. Specifically, we measure both the \textit{instance distortion} (relative to the underlying utility profile $\uprof$) and the \textit{worst-case distortion} (over all utility profiles consistent with the observed preferences $\prof$). Due to space constraints, we report results for the latter, and only as a function of the dimension 
$d$ in the main text. 
Results varying 
$n,m$, running times, and instance distortion, are deferred to Appendix \ref{app:additional_experiment}.

\smallskip
\noindent
\textbf{Datasets.} We consider the MovieLens 100K~\citep{harper2015movielens} and abortion opinion survey~\citep{fish2024generative} datasets.
MovieLens contains 100K movie ratings from 1000 users (voters) over 1700 movies.
We translate it into our setting via approximate matrix decomposition for each dimension $d$, adding the nonnegativity and normalization constraints.
We then randomly subsample 20 movie preference votes for each dimension $d$ with $n=100$ and $m=25$.
The abortion opinions dataset includes ratings from 100 individuals on 5 different opinion statements. For this setting, we embed the candidates using Matryoshka embeddings~\cite{kusupati2022matryoshka}, followed by dimensionality reduction via principal component analysis (PCA). The embedding choice for either dataset is not essential to our method and is intended only as a representative example.

\begin{figure}[h]
    \centering
    \includegraphics[width=1\linewidth]{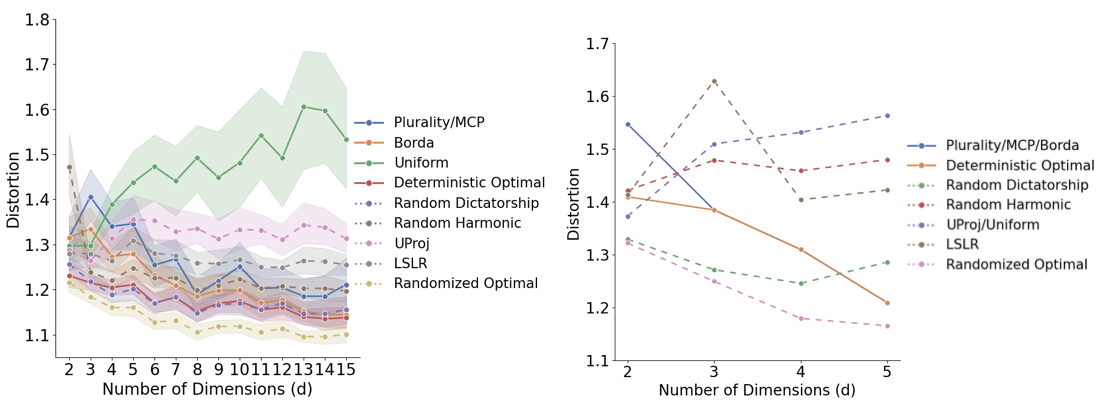}
    \caption{Instance distortion $\dist(f,\prof)$. Left: MovieLens; Right: Abortion Opinion Survey}
\label{fig:distortion_bound_both}
\end{figure}

\smallskip
\noindent\textbf{Results.} 
Figure~\ref{fig:distortion_bound_both} presents instance-based distortion bounds for MovieLens (left) and abortion opinion survey (right).
Our first observation is that, as we expect, the instance-optimal approaches we propose outperform all alternatives, with improvement appearing to increase with $d$.
This is true for both deterministic and randomized rules.
Second, and somewhat surprisingly, we find that worst-case distortion bounds \emph{decrease with dimension} $d$, despite the fact that distortion lower bounds increase in $d$.
Thus, as the dimension of candidate representations increases, using rankings in place of utilities becomes empirically less consequential.

\section{Conclusion}
We initiate the study of distortion in the setting of linear social choice, which we anticipate will play a role in the integration of social choice into neural and representation-based models.
For deterministic rules we introduce the maximum coordinate plurality (MCP) rule, which we prove obtains worst-case distortion $\mathcal{O}(d^3)$ and is within $\Theta(d)$ of optimal.
Since MCP requires access to candidate positions, we leave open whether $d$-dependent distortion is even \emph{possible} for deterministic rules without such access.

For randomized rules, we construct the \emph{linear stable lotteries rule}, and show it attains asymptotically optimal worst-case distortion $\Theta(\sqrt{d})$ when candidate locations are known.
We show stable lotteries yield $O(d)$ distortion even when candidate locations are unknown, leaving unresolved the optimal distortion in this setting.
We also establish preliminary results for $\elltwo$ voter and candidate normalization.

Beyond the utilitarian welfare objective, future work might also explore \emph{Nash welfare} or \emph{proportional fairness}, as studied by ~\citet{ebadian2024optimized}, or worst-case \emph{regret} minimization \cite{caragiannis2017subset,kahng2022worst}.
Finally, given the connection between linear social choice and reward learning, we plan to investigate the robustness of various rules under both learning errors and strategic voters.

\clearpage
\bibliography{aaai2026}

\begin{thebibliography}{24}
\providecommand{\natexlab}[1]{#1}

\bibitem[{Anshelevich et~al.(2021)Anshelevich, Filos-Ratsikas, Shah, and Voudouris}]{anshelevich2021distortion}
Anshelevich, E.; Filos-Ratsikas, A.; Shah, N.; and Voudouris, A.~A. 2021.
\newblock {Distortion in Social Choice Problems: The First 15 Years and Beyond}.
\newblock In \emph{Proceedings of the Thirtieth International Joint Conference on Artificial Intelligence (IJCAI-21)}.

\bibitem[{Aziz(2020)}]{aziz2020justifications}
Aziz, H. 2020.
\newblock {Justifications of Welfare Guarantees Under Normalized Utilities}.
\newblock \emph{{ACM SIGecom Exchanges}}.

\bibitem[{{Berker} et~al.(2025){Berker}, {Casacuberta}, {Robinson}, {Ong}, {Conitzer}, and {Elkind}}]{berker25independence}
{Berker}, R.~E.; {Casacuberta}, S.; {Robinson}, I.; {Ong}, C.; {Conitzer}, V.; and {Elkind}, E. 2025.
\newblock {From Independence of Clones to Composition Consistency: A Hierarchy of Barriers to Strategic Nomination}.
\newblock In \emph{Proceedings of the 26th ACM Conference on Economics and Computation}.

\bibitem[{Bhaskar, Dani, and Ghosh(2018)}]{bhaskar2018truthful}
Bhaskar, U.; Dani, V.; and Ghosh, A. 2018.
\newblock {Truthful and Near-Optimal Mechanisms for Welfare Maximization in Multi-winner Elections}.
\newblock In \emph{Proceedings of the AAAI Conference on Artificial Intelligence}.

\bibitem[{{Boutilier} et~al.(2015){Boutilier}, {Caragiannis}, {Haber}, {Lu}, {Procaccia}, and {Sheffet}}]{boutilier15optimal}
{Boutilier}, C.; {Caragiannis}, I.; {Haber}, S.; {Lu}, T.; {Procaccia}, A.~D.; and {Sheffet}, O. 2015.
\newblock {Optimal Social Choice Functions: A Utilitarian View}.
\newblock \emph{Artificial Intelligence}.

\bibitem[{Caragiannis et~al.(2017)Caragiannis, Nath, Procaccia, and Shah}]{caragiannis2017subset}
Caragiannis, I.; Nath, S.; Procaccia, A.~D.; and Shah, N. 2017.
\newblock {Subset Selection via Implicit Utilitarian Voting}.
\newblock \emph{Journal of Artificial Intelligence Research}.

\bibitem[{Cheng et~al.(2020)Cheng, Jiang, Munagala, and Wang}]{cheng2020group}
Cheng, Y.; Jiang, Z.; Munagala, K.; and Wang, K. 2020.
\newblock {Group Fairness in Committee Selection}.
\newblock \emph{ACM Transactions on Economics and Computation (TEAC)}.

\bibitem[{Dantzig and Wolfe(1960)}]{dantzig1960decomposition}
Dantzig, G.~B.; and Wolfe, P. 1960.
\newblock {Decomposition Principle for Linear Programs}.
\newblock \emph{Operations Research}.

\bibitem[{Ebadian et~al.(2024{\natexlab{a}})Ebadian, Filos-Ratsikas, Latifian, and Shah}]{ebadian2024computational}
Ebadian, S.; Filos-Ratsikas, A.; Latifian, M.; and Shah, N. 2024{\natexlab{a}}.
\newblock {Computational Aspects of Distortion}.
\newblock In \emph{The 23rd International Conference on Autonomous Agents and Multiagent Systems}.

\bibitem[{Ebadian et~al.(2024{\natexlab{b}})Ebadian, Kahng, Peters, and Shah}]{ebadian2024optimized}
Ebadian, S.; Kahng, A.; Peters, D.; and Shah, N. 2024{\natexlab{b}}.
\newblock {Optimized Distortion and Proportional Fairness in Voting}.
\newblock \emph{{ACM Transactions on Economics and Computation}}.

\bibitem[{{Fish} et~al.(2024){Fish}, G{\"o}lz, {Parkes}, {Procaccia}, {Rusak}, {Shapira}, and {W{\"u}thrich}}]{fish2024generative}
{Fish}, S.; G{\"o}lz, P.; {Parkes}, D.~C.; {Procaccia}, A.~D.; {Rusak}, G.; {Shapira}, I.; and {W{\"u}thrich}, M. 2024.
\newblock {Generative Social Choice}.
\newblock In \emph{Proceedings of the 25th ACM Conference on Economics and Computation}.

\bibitem[{{Ge} et~al.(2024){Ge}, {Halpern}, {Micha}, {Procaccia}, {Shapira}, {Vorobeychik}, and {Wu}}]{ge2024axioms}
{Ge}, L.; {Halpern}, D.; {Micha}, E.; {Procaccia}, A.~D.; {Shapira}, I.; {Vorobeychik}, Y.; and {Wu}, J. 2024.
\newblock {Axioms for AI Alignment from Human Feedback}.
\newblock In \emph{Advances in Neural Information Processing Systems}.

\bibitem[{G{\"o}lz, Haghtalab, and Yang(2025)}]{golz2025distortion}
G{\"o}lz, P.; Haghtalab, N.; and Yang, K. 2025.
\newblock {Distortion of AI Alignment: Does Preference Optimization Optimize for Preferences?}
\newblock \emph{arXiv preprint}.
\newblock ArXiv:2505.23749.

\bibitem[{Harper and Konstan(2015)}]{harper2015movielens}
Harper, F.~M.; and Konstan, J.~A. 2015.
\newblock {The MovieLens Datasets: History and Context}.
\newblock \emph{ACM Transactions on Interactive Intelligent Systems (TiiS)}.

\bibitem[{Kahng and Kehne(2022)}]{kahng2022worst}
Kahng, A.; and Kehne, G. 2022.
\newblock {Worst-Case Voting When the Stakes Are High}.
\newblock In \emph{Proceedings of the AAAI Conference on Artificial Intelligence}.

\bibitem[{Kusupati et~al.(2022)Kusupati, Bhatt, Rege, Wallingford, Sinha, Ramanujan, Howard-Snyder, Chen, Kakade, Jain, and Others}]{kusupati2022matryoshka}
Kusupati, A.; Bhatt, G.; Rege, A.; Wallingford, M.; Sinha, A.; Ramanujan, V.; Howard-Snyder, W.; Chen, K.; Kakade, S.; Jain, P.; and Others. 2022.
\newblock {Matryoshka Representation Learning}.
\newblock In \emph{Advances in Neural Information Processing Systems}.

\bibitem[{Liang(2025)}]{liang25ai}
Liang, A. 2025.
\newblock {Artificial Intelligence Clones}.
\newblock In \emph{Proceedings of the 26th ACM Conference on Economics and Computation, EC 2025}.

\bibitem[{{Lin}(2024)}]{Lin:GrindrAI}
{Lin}, B. 2024.
\newblock {Grindr Aims to Build the Dating World's First AI 'Wingman'}.
\newblock \emph{{The Wall Street Journal}}.
\newblock October 5.

\bibitem[{Ngatchou, Zarei, and El-Sharkawi(2005)}]{ngatchou2005pareto}
Ngatchou, P.; Zarei, A.; and El-Sharkawi, A. 2005.
\newblock {Pareto Multi Objective Optimization}.
\newblock In \emph{Proceedings of the 13th International Conference on Intelligent Systems Application to Power Systems}. IEEE.

\bibitem[{Pennock et~al.(2000)Pennock, Horvitz, Giles, and Others}]{pennock2000social}
Pennock, D.~M.; Horvitz, E.; Giles, C.~L.; and Others. 2000.
\newblock {Social Choice Theory and Recommender Systems: Analysis of the Axiomatic Foundations of Collaborative Filtering}.
\newblock In \emph{Proceedings of the AAAI Conference on Artificial Intelligence}.

\bibitem[{Procaccia and Rosenschein(2006)}]{procaccia06distortion}
Procaccia, A.~D.; and Rosenschein, J.~S. 2006.
\newblock {The Distortion of Cardinal Preferences in Voting}.
\newblock In \emph{Proceedings of the International Workshop on Cooperative Information Agents}. Springer.

\bibitem[{{Procaccia}, {Schiffer}, and {Zhang}(2025)}]{procaccia2025clone}
{Procaccia}, A.~D.; {Schiffer}, B.; and {Zhang}, S. 2025.
\newblock {Clone-Robust AI Alignment}.
\newblock In \emph{Proceedings of the International Conference on Machine Learning}. PMLR.

\bibitem[{{Tan}(2024)}]{NYT:ModernLove:DatingAppsSuck}
{Tan}, E. 2024.
\newblock {Dating Apps Suck. A.I. Clones Are Making Them Even Weirder.}
\newblock \emph{{The New York Times}}.
\newblock December 11.

\bibitem[{Zhu, Jordan, and Jiao(2023)}]{zhu2023principled}
Zhu, B.; Jordan, M.; and Jiao, J. 2023.
\newblock {Principled Reinforcement Learning with Human Feedback from Pairwise or K-wise Comparisons}.
\newblock In \emph{Proceedings of the International Conference on Machine Learning}. PMLR.

\end{thebibliography}

\clearpage
\onecolumn

\newpage

\appendix

\section{Omitted Proofs}
\label{app:omitted-proofs}

\randomlb*
\begin{proof}[Proof of \Cref{thm:random-LB-ell1}]
    This is essentially a presentation of the proof of \citet{boutilier15optimal}.
    
    Without loss of generality, assume that $\sqrt{d}$ divides $n$. 
    Partition the set of agents into $\sqrt{d}$ equal-size groups $N_1, \ldots, N_{\sqrt{d}}$.  
    Consider the following preference profile: for each $k = 1, \ldots, \sqrt{d}$, all agents in $N_k$ rank candidate $c_k$ first (the remainder of their rankings can be arbitrary). 
    
    By the pigeonhole principle, for any randomized social choice function $f$, there exists a $k^* \in \{1, \ldots, \sqrt{d}\}$ such that $\Pr[f(\prof) = c_{k^*}] \leq \frac{1}{\sqrt{d}}$.
    
    Now, construct the following instance (which is \emph{unknown} to the social choice rule):  
    Let $c_{k^*} = c_1$, and for $k \in [\sqrt{d}]$, let $c_k = e_k$ (the $k$-th standard basis vector). For all other candidates $j > \sqrt{d}$, let $c_j = e_2$. Assign all voters in $N_{k^*}$ the vector $v_1 = e_1$, and all voters in other groups the vector $\mu$. It is easy to check that these voters and candidates satisfy all constraints imposed in the model section.
    
    For this configuration, the average social welfare is maximized by selecting $c_{k^*}$, yielding
    \[
    \sw_1 = \frac{1}{\sqrt{d}} + \frac{1}{d}\left(1 - \frac{1}{\sqrt{d}}\right).
    \]
    For any other candidate, the average social welfare is 
    \[
    \sw_2 < \frac{1}{d}.
    \]
    
    Thus, the welfare-maximizing candidate is $c_{k^*}$. However, since $\Pr[f(\prof) = c_{k^*}] \leq \frac{1}{\sqrt{d}}$, the expected social welfare under $f$ is at most
    \[
    \mathbb{E}[\sw] < \sw_1 \cdot \frac{1}{\sqrt{d}} + \sw_2 \cdot \left(1 - \frac{1}{\sqrt{d}}\right).
    \]
    Therefore, the distortion is at least
    \[
    \dist(f) > \frac{\sw_1}{\sw_1 \cdot \frac{1}{\sqrt{d}} + \sw_2 \cdot \left(1 - \frac{1}{\sqrt{d}}\right)}.
    \]
    Plugging in $\sw_1 = \frac{1}{\sqrt{d}} + \frac{1}{d}\left(1 - \frac{1}{\sqrt{d}}\right)$ and $\sw_2 = \frac{1}{d}$, we find
    \[
    \dist(f) \in \Omega(\sqrt{d}).
    \]
    Thus, the distortion is $\Omega(\sqrt{d})$, as claimed.
\end{proof}

\randomdictator*
\begin{proof}[Proof of \Cref{thm:RD-distortion-ell1}]
    The $\Omega(d)$ lower bound for the distortion of $\rd$ in the unit-sum setting is straightforward; we reproduce it here for completeness. 
    Consider $m = d$, with each candidate embedded to a unit basis vector; i.e. $C = [d]$ and $c_i = e_i$ for each $i$.
    We will consider $d-1$ equal-size groups of voters $V_i$ for $i \in [d-1]$. 
    The voters in $V_i$ have embedding $v = (e_i + e_d)/2$ (or $v = e_i (1/2 + \epsilon) + e_d (1/2 - \epsilon)$ for some small $\epsilon > 0$, if ties must be avoided), and so their rankings will be $c_i \succ_v c_d \succ_v \dots$. 
    On such an instance, $\rd$ chooses uniformly from $[d-1]$ and $\sw(\rd)= \frac{n}{2(d-1)}$, while the welfare optimizer is $c_d$ with $\sw(c_d) = \frac{n}{2}$. 
    Therefore on this profile $\prof$ we have 
    \[
        \dist(\rd, \prof) = \frac{1}{d-1} = \Omega(d).
    \]

    We now turn to upper bounding the distortion of $\rd$.
	For each $v \in V$, let $c_v$ denote the top-ranked candidate for $v$. 
	By Lemma~\ref{lem:voter-util-LB} we have that $u_v(c_v) = v^T c_v \geq 1/d$.
	So by an averaging argument over the $d$ dimensions, there is some $\ell \in [d]$ such that $v^\ell c_v^\ell \geq 1/d^2$.
	Since $v, c_v \in \Delta_d$, this implies that $v^\ell \geq 1/d^2$ and $c_v^\ell \geq 1/d^2$.
		
    Fix one such coordinate $\ell_v$ for each $v$ and partition the voters $V$ into groups $V_1, \ldots, V_d$ according to their $\ell_v$.
    Let $p_\ell \defeq |V_\ell|/n$ encode the proportion of voters in each group, and observe that $p \in \Delta_d$.
	We can now lower bound the expected welfare from $\rd$ by observing that if the dictator is chosen from group $V_i$, all other voters in $V_i$ will also obtain some welfare from their choice.
	\begin{align*}
		E_{c\sim RD}\left[ \sw(c) \right] &= \frac{1}{n} \sum_{v \in V} \sw(c_v) \\
		&= \frac{1}{n} \sum_{v \in V} \sum_{v' \in V} u_{v'}(c_v) \\
		&\geq \frac{1}{n} \sum_{v \in V} \sum_{v' \in V_{\ell_v}} v'^T c_v \\
		&\geq \frac{1}{n} \sum_{v \in V} \sum_{v' \in V_{\ell_v}} (v')_{\ell_v} (c_v)_{\ell_v}.
	\end{align*}
	From here we will rearrange the sum to range over the groups $v^\ell$. Observe that $v^\ell c_v^\ell \geq 1/d^2$ implies that $c_v^\ell \geq \frac{1}{v^\ell d^2}$. Then continuing,
	\begin{align*}
		&= \frac{1}{n} \sum_{\ell} \sum_{v,v' \in V_{\ell}} v^\ell c^{v'}_\ell \\
		&\geq \frac{1}{n} \sum_{\ell} \sum_{v,v' \in V_{\ell}} v^\ell \frac{1}{d^2 v'_\ell} \\
		&\geq \frac{1}{n} \sum_{\ell} \sum_{v,v' \in V_{\ell}: v^\ell \geq v'_\ell} v^\ell \frac{1}{d^2 v'_\ell} \\
		&\geq \frac{1}{n} \sum_{\ell} \sum_{v,v' \in V_{\ell}: v^\ell \geq v'_\ell} v^\ell \frac{1}{d^2 v^\ell} \\
		&\geq \frac{1}{n d^2} \sum_\ell \frac{1}{2} |v^\ell|^2 \\
		&= \frac{n}{2 d^2} p^T p \\
		&\geq \frac{n}{2 d^3}.
	\end{align*}
	The claim follows from the observation that the maximum welfare is at most $n$.
\end{proof}
We conclude by remarking that while the analysis of $\rd$ in this proof in some sense proceeds on a dimension-by-dimension basis, the rule is itself agnostic to (and indeed has no knowledge of) these coordinates or the embeddings of $V$ and $C$. 
This is in contrast to $\mcp$, where both the rule and its analysis proceed along these lines.

\harmonic*
\begin{proof}[Proof of \Cref{thm:harmonic-vs-RD}]
	We start with the first claim. 
    Consider the instance for which $d=2$, all voters are $v=(1,0)$ and perfectly aligned with $c=(1,0)$, and all other candidates are $c'=(0,1)$. 
    Then $\sw(\hr) = \frac{n}{2 H_m}$, where $H_m = \log m + \Theta(1)$ is the $m^{th}$ harmonic number. 
    On the other hand $\sw(c) = n$, so for this instance family $\{\uprof_m\}$ we have $\dist(\hr, \uprof_m) = \Theta(\log^{-1} m) \rightarrow 0$ as $m$ grows large.

    We now address the second. 
    Observe that on any profile $\prof$ induced by a utility profile instance $\uprof$, and for every alternative $a$, $\Pr_{c \sim \hr(\prof)}[c = a] \geq \Pr_{c \sim \rd(\prof)}[c = a] /H_m$.
    Therefore 
    \begin{align}
        \sw(\hr(\prof)) &= \sum_a \Pr_{c \sim \hr(\prof)}\left[ c = a \right] \sw(a)  \notag \\
        &\geq \frac{1}{H_m} \sum_a \Pr_{c \sim \rd(\prof)}\left[ c = a \right] \sw(a) \notag \\
        &= \frac{\sw(\rd(\prof))}{H_m}.\notag
    \end{align}
    The claim follows from $H_m \leq \log m + 1$.
\end{proof}

\purestablelottery*
\begin{proof} [Proof of \Cref{thm:pure-stable-lottery-ell1}]

	As in the proof of \Cref{thm:stable-lottery-ell1}, for any committee $W$ of size $k$ and any $a \in W$, let $S_a(W) \subseteq V$ denote the voters for which $a \succ_v W$, and $\overline S_a(W)$ its complement. 
	Then via identical derivation we obtain 
	\begin{align}
		\sum_{v \in V} \sum_{c \in W} u_v(c) 
		&\geq \sum_{v\in V} u_v(c^*)  - |\{S_a(W)\}|.
	\end{align}
	Taking the expectation over a stable lottery $\cW$ and applying \Cref{def:stable-lotteries} with $k=2d$, we obtain
	\begin{align}
		\ev_{W \sim \cW} \left[  \sum_{c \in W} \sw(c) \right]
		&\geq \sw(c^*)  - \ev_{W \sim \cW} \left[ | \{S_a(W)\} | \right] \notag \\
		&\geq \sw(c^*)  - \frac{n}{2d} \notag \\
		&\geq \frac{1}{2} \cdot \sw(c^*), \label{eq:stab-lot-vs-welfare-alone}
	\end{align}
	where the last step follows from the fact that $\sw(c^*)\geq n/d$ by \Cref{thm:middle-dist-ell1}.
	Applying \eqref{eq:stab-lot-vs-welfare-alone} and recalling our choice of $k = 2d$, we have 
	\begin{align}
		k \cdot \ev_{c \sim \pslr } \left[ \sw(c) \right] &= \ev_{W \sim \cW} \left[  \sum_{c \in W} \sw(c) \right] \notag \\
		&\geq \frac{1}{2} \cdot \sw(c^*), \notag
        \intertext{and so}
		\ev_{c \sim \lslr } \left[ \sw(c) \right]  &\geq \frac{1}{4 d} \cdot \sw(c^*).  \notag 
	\end{align}
	This proves the claim.
\end{proof}

\section{Pseudocode for Implementations of Instance-Optimal Rules }\label{app:pseudocode}

\begin{algorithm}
\caption{Deterministic Distortion-Optimal Rule (Column Generation)}
\begin{algorithmic}[1]
\STATE \textbf{INPUT:} Candidate set \( C \), preference profile \( \prof \), tolerance \( \varepsilon \)
\FOR{each candidate \( c_k \in C \)}
    \STATE Initialize \( \mathcal{V} \gets \) small subset of feasible average voters in \( \mathcal{F} \)
    \REPEAT
        \STATE Solve the following LP to obtain \( \hat{\beta}_k \):
        \[
        \begin{aligned}
            \max_{\beta} \quad & \beta \\
            \text{s.t.} \quad &
            \langle c_k, \bar{v} \rangle \ge \beta \cdot \max_{c \in C} \langle c, \bar{v} \rangle 
            \quad \forall \bar{v} \in \mathcal{V}
        \end{aligned}
        \]
        \STATE Call \textsc{SeparationOracle}\((c_k, \hat{\beta}_k, \prof)\) to find a violated \( \bar{v} \in \mathcal{F} \)
        \IF{a violated \( \bar{v} \) is found}
            \STATE ADD \( \bar{v} \) to \( \mathcal{V} \)
        \ELSE
            \STATE \textbf{BREAK}
        \ENDIF
    \UNTIL{convergence}
    \STATE Store worst-case distortion bound \( \hat{\beta}_k \)
\ENDFOR
\RETURN Candidate \( c_k \in C \) with largest \( \hat{\beta}_k \)
\end{algorithmic}
\end{algorithm}

\begin{algorithm}[ht]
\caption{\textsc{SeparationOracle}$(\hat{c}, \hat{\beta}, \prof)$}
\begin{algorithmic}[1]
\FOR{each candidate $c \in C$}
    \STATE Let $d_c \gets \hat{c} - \hat{\beta} \cdot c$
    \STATE Solve LP:
    \[
    \begin{aligned}
    \min_{\bar{v} \in \mathcal{F}} \quad & \langle d_c, \bar{v} \rangle 
    \end{aligned}
    \]
    \IF{optimal value $< 0$}
        \STATE \textbf{Return} violating $\bar{v}$
    \ENDIF
\ENDFOR
\STATE \textbf{Return} \textsc{None} 
\end{algorithmic}
\end{algorithm}

\begin{algorithm}[ht]
\caption{Column Generation for Randomized Distortion-Optimal Rule}
\begin{algorithmic}[1]
\STATE \textbf{INPUT:} Candidate set \( C \), preference profile \( \prof \), tolerance \( \varepsilon \)
\STATE \textbf{INITIALIZE:} Constraint set \( \mathcal{V} \gets \) small subset of \( \mathcal{F} \) (e.g., one feasible \( \bar{v} \))
\REPEAT
    \STATE Solve the following LP to obtain \( (\hat{\beta}, \hat{p}) \):
    \[
    \begin{aligned}
    \max_{\beta,\, p} \quad & \beta \\
    \text{s.t.} \quad &
    \sum_{i=1}^{m} p_i \langle c_i, \bar{v} \rangle 
    \ge \beta \cdot \max_{c \in C} \langle c, \bar{v} \rangle 
    \quad \forall \bar{v} \in \mathcal{V} \\
    & \sum_{i=1}^{m} p_i = 1,\quad p_i \ge 0 \quad \forall i
    \end{aligned}
    \]
    \STATE CALL \textsc{SeparationOracle}\(( \hat{p}^TC,\hat{\beta}, \prof)\) to find a violated constraint \( \bar{v} \in \mathcal{F} \)
    \IF{a violated \( \bar{v} \) is returned}
        \STATE ADD \( \bar{v} \) to \( \mathcal{V} \)
    \ELSE
        \STATE \textbf{BREAK}
    \ENDIF
\UNTIL{convergence}
\STATE \textbf{RETURN:} Randomized rule defined by \( \hat{p} \)
\end{algorithmic}
\end{algorithm}

Instead of performing binary search independently for each pair of candidates in the deterministic instance-optimal rule, we can adopt a structure similar to that used in the randomized rule. This still requires $O(m^2)$ different programs, but now, for each candidate, we compute its maximum distortion against all others directly. In this formulation, the separation oracle is invoked repeatedly with a fixed feasible region but varying objective functions. This allows us to use warm-starting to significantly accelerate computation.

\section{Results Under $\elltwo$ Normalization}
\label{app:ell-two-normalization}

When voter and candidate vectors are $\elltwo$-normalized, meaning that $\|v\|_2 = \|c\|_2 = 1$ utilities are given by $u_v(c) = v^Tc = \cos(\theta_{vc})$, where $\theta_{vc}$ is the angle between $v$ and $c$.

A critical difference between these models is a voter's favorite candidate: under $\ellone$ normalization, any voter $v$ has a favorite candidate in $\{e_1, \ldots, e_d\}$, where $e_i$ is the standard basis vector for coordinate $i \in [d]$.
By contrast, under $\elltwo$ normalization the favorite candidate of any voter $v$ is itself.

\begin{table}[ht]
    \centering
    \begin{tabular}{c c c}
        $C$ Known? & Randomized & Deterministic \\
        \hline
        \hline
        \multirow{2}{*}{Known} & $O(d^{1/2})$ \:\: (Thm.~\ref{thm:stable-lottery-ell2}) & $O(d^2)$ \:\: (Thm.~\ref{thm:MCP-dist-ell2}) \\
        & $\Omega(d^{1/4})$ \:\: (Thm.~\ref{thm:random-LB-ell2}) & $\Omega{(d^{3/2})}$ \:\:(Thm.~\ref{thm:det-LB-ell2}) \\
        \hline
        Unknown & $O(d)$ \:\:  (Thm.~\ref{thm:pure-stable-lottery-ell2}) & - \\
    \end{tabular}
    \caption{Preliminary distortion bounds in the $\elltwo$-normalized linear social choice setting. Again the `$C$ unknown' setting inherits lower bounds from the `$C$ known' setting.}
    \label{tab:distortion-UBs-ell2}
\end{table}

\subsection{Impossibilities}

Throughout let $\mu_2=(\frac{1}{\sqrt{d}},\dots, \frac{1}{\sqrt{d})})$ denote the uniform $\elltwo$-normalized vector in $\mathbb{R}^d_{\geq 0}$.

\begin{theorem}
\label{thm:det-LB-ell2}
    Suppose $m, n \geq d$. 
    Then there is a profile $\prof$ on which every deterministic rule $f$ has distortion $\Omega(d^{3/2})$.
\end{theorem}
\begin{proof}
    We follow the helpful presentation by \citet{anshelevich2021distortion} of the $\Omega(m^2)$ lower bound of \citet{caragiannis2017subset} on the distortion of any deterministic rule in the unit-sum setting.

    Voters are organized into two equal-size groups $V_a$ and $V_b$, where each voter has a unique preferred candidate $c_v$, all voters $v \in V_a$ rank $a$ second, and all voters $v \in V_b$ rank $b$ second.
    Suppose that  $d = m$ and (as usual) that the embeddings of the candidates form the standard basis for $\mathbb{R}^d$.
    In this case, ranking suffixes may be completed arbitrarily; this defines our profile $\prof$.

    What can our deterministic rule do? 
    If it outputs a candidate in $\{a,b\}$ then suppose without loss of generality that $f(\prof) = b$ and $a$ is the welfare maximizer.
    If all $v \in V_b$ are embedded as $v = c_v$, then $\sw(f(\prof)) = 0$, while the maximizer can be anything supported by group $V_a$, leading to unbounded distortion.

    On the other hand, suppose that $f(\prof) = c_v$, some voter's favorite candidate.
    Suppose without loss of generality $v \in V_b$ and that $a$ is the welfare maximizer, and all $v' \in A$ are embedded as $v' = (a + c_{v'})/2$, so that $\sw(a) = \Omega(n)$.
    If $v$ is essentially indifferent and embedded at $\mu_2$, but no other voters place any value on $c_v$, then $\sw(c_v) = v^T c_v = d^{-1/2}$.
    In this case, we have 
    \[
        \dist(f(\prof)) = \Omega(d^{3/2}).
    \]
    Since in all other cases $f$ suffers infinite distortion, this is the best it can attain on such a profile $\prof$.
\end{proof}

In the same spirit as Theorem~\ref{thm:random-LB-ell1}, we have a lower bound for any randomized rule that is agnostic to candidates' positions.
\begin{theorem}
\label{thm:random-LB-ell2}
    Suppose $m \geq d$ and $n \geq d^{1/4}$. Then there exists a preference profile $p$ such that for any randomized social choice function $f$, we have $\dist(p, f(p)) = \Omega(d^{1/4})$.
\end{theorem}

\begin{proof}
Without loss of generality, assume that $d^{1/4}$ divides $n$. Partition the set of agents into $d^{1/4}$ equally sized groups, $N_1, \ldots, N_{d^{1/4}}$. Consider the following preference profile: for each $k = 1, \ldots, d^{1/4}$, all agents in $N_k$ rank candidate $c_k$ first (the remainder of their rankings can be arbitrary). 

By the pigeonhole principle, for any randomized social choice function $f$, there exists $k^* \in \{1, \ldots, d^{1/4}\}$ such that $\Pr[f(p) = c_{k^*}] \leq \frac{1}{d^{1/4}}$.

Now, construct the following instance (which is \emph{unknown} to the social choice rule):  
Let $c_{k^*} = c_1$, and for $k \in [d^{1/4}]$, let $c_k = e_k$ (the $k$-th standard basis vector). For all other candidates $j > d^{1/4}$, let $c_j = e_2$. Assign all voters in $N_{k^*}$ the vector $v_1 = e_1$, and all voters in other groups the vector $\mu_2$. It is easy to check that these voters and candidates satisfy all constraints imposed in the model section.

For this configuration, the average social welfare is maximized by selecting $c_{k^*}=e_1$, yielding
\[
\sw_1 = \frac{1}{d^{1/4}} + \frac{1}{\sqrt{d}}\left(1 - \frac{1}{d^{1/4}}\right).
\]
For any other candidate, the average social welfare is \[
\sw_2 < \frac{1}{\sqrt{d}}.
\]

Thus, the welfare-maximizing candidate is $c_{k^*}$. However, since $\Pr[f(p) = c_{k^*}] \leq \frac{1}{d^{1/4}}$, the expected social welfare under $f$ is at most
\[
\mathbb{E}[\sw] < \sw_1 \cdot \frac{1}{d^{1/4}} + \sw_2 \cdot \left(1 - \frac{1}{d^{1/4}}\right).
\]
Therefore, the distortion is at least
\[
\dist(p, f(p)) > \frac{\sw_1}{\sw_1 \cdot \frac{1}{d^{1/4}} + \sw_2 \cdot \left(1 - \frac{1}{d^{1/4}}\right)}.
\]
Plugging in $\sw_1 = \frac{1}{d^{1/4}} + \frac{1}{\sqrt{d}}\left(1 - \frac{1}{d^{1/4}}\right)$ and $\sw_2 = \frac{1}{\sqrt{d}}$, we find
\[
\dist(p, f(p)) \in \Omega(d^{1/4}).
\]
Thus, the distortion is $\Omega(d^{1/4})$, as claimed.
\end{proof}

\begin{observation}
    Random dictatorship has worst-case distortion $\Omega(d)$ under $\elltwo$ normalization.
\end{observation}
\begin{proof}
    The instance is the same as in the proof of \Cref{thm:RD-distortion-ell1}, up to constant factors due to the renormalization of the sparse utilities.
    Consider $C = \{e_1, \ldots, e_d\}$ and $d-1$ equal-size groups of voters.
    Again voters in group $i$ are nearly indifferent between $c_i$ and $c_d$, but rank $c_i$ above $c_d$, and garner $0$ utility from all other candidates.
    Then the welfare-maximizing candidate is $c_d$ with welfare $\sw(c^*) = \Omega(n)$, while $\sw(\rd) = O(\frac{n}{d})$.
\end{proof}

\subsection{Positive Results}
\label{sec:upper-bounds-ell2}

We start with an analog of \Cref{lem:voter-util-LB}.
\begin{lemma}
\label{lem:voter-util-LB-ell2}
    For any set of candidates $C=\{c_j\}_{j \in [m]}$ and any voter $v \in \cone(C)$,
    \[
        \max_{c \in C} u_v(c) \geq d^{-1/2}.
    \]
\end{lemma}
\begin{proof}[Proof of \Cref{lem:voter-util-LB-ell2}]
    We will denote the maximum utility for $v$ by $\hat u_v \defeq \max_{c \in C} u_v(c)$; our goal will then be to provide a lower bound on $\hat u_v$.
    Since $v$ and $c \in C$ are $\elltwo$-normalized, this coincides with lower bounding the cosine similarity between $v$ and the most similar $c \in C$.
    
    To begin, since $v \in \cone(C)$ we may write $v = \sum_c \alpha_c c$ for some collection of $\alpha_c \geq 0$. 
    Then to begin we have
    \[
        1 = v^T v = \sum_c \alpha_c v^\top c \leq \hat u_v \sum_c \alpha_c,
    \]
    and so $\hat u_v \geq \left(\sum_c \alpha_c\right)^{-1}$.
    Next, we use the standard $\ellone$-$\elltwo$ norm inequalities, which state that for all $x \in \mathbb{R}^d$, $\|x\|_2 \leq \|x\|_1 \leq \|x\|_2 \sqrt{d}$.
    Therefore $1 = \|c\|_2 \leq \|c\|_1$, and so
    \[
        \sum_c \alpha_c = \sum_c \alpha_c \|c\|_2  \leq \sum_c \alpha_c \|c\|_1 = \|v\|_1 \leq \sqrt{d}.
    \]
    Combining gives us $\hat u_v \geq d^{-1/2}$, as claimed.
\end{proof}
We conclude by observing that \Cref{lem:voter-util-LB-ell2} is tight when $v = \mu_2$ and $C$ is the standard basis of $\mathbb{R}^d$. 

As a consequence, we obtain an improved distortion bound for $\mcp$ under $\elltwo$ normalization.
\begin{theorem}\label{thm:MCP-dist-ell2}
   The distortion of Max Coordinate Plurality is $O(d^2)$ under $\elltwo$ normalization of voters and candidates.
\end{theorem}
\begin{proof}
    Following the argument of \Cref{thm:MCP-dist-ell1}, know that each $v \in V$ has a coordinate with value at least $d^{-1/2}$. 
    Therefore there is some $c \in \hat{C}$ with value at least $d^{-1/2}$ for this same coordinate by the cone assumption, and so $\max_{c' \in \hat C} v^T c' \geq v^T c \geq 1/d$. 

    Since the plurality winner in $\hat C$ is supported by at least $n/d$ voters, the winner $c^* \in \hat C$ chosen by MCP confers welfare 
    \[
        \sw(c^*) \geq \frac{n}{d}\cdot \frac{1}{d},
    \]
    and the claim follows.
\end{proof}

\begin{theorem}\label{thm:RD-distortion-ell2}
   Random Dictatorship has distortion at most $O(d^{5/2})$ under $\elltwo$ normalization of voters and candidates.
\end{theorem}
\begin{proof}
    We follow the proof of \Cref{thm:RD-distortion-ell1}. 

    For each agent $v \in V$, let $c_v$ denote the top-ranked alternative for $v$.
    By \Cref{lem:voter-util-LB-ell2} we have that $u_v(c_v) = v^T c_v \geq d^{-1/2}$.
    So by averaging, there is some $\ell \in [d]$ such that $v^\ell c_v^\ell \geq d^{-3/2}$. %
    Since $v, c_v \in S^d_+$, this implies $v^\ell \geq d^{-3/2}$ and  $c_v^\ell \geq d^{-3/2}$.
    Fix one such coordinate $\ell_v$ for each $v$ and partition $V$ into $V_1, \ldots, V_d$ according to $\ell_v$.
    Let $p_\ell = |v^\ell|/n$, and observe that $p \in \Delta_d$.
    Let $c^* = \arg\max_{c \in C} \sum_{v \in V} u_v(c)$ be the welfare-maximizing candidate. 
    Clearly,
    \[
        \sw(c^*) \leq n.
    \]

    On the other hand, 
    \begin{align*}
        E_{c\sim RD}\left[ \sw(c) \right] &= \frac{1}{n} \sum_{v \in V} \sw(c_v) \\
        &= \frac{1}{n} \sum_{v \in V} \sum_{v' \in V} u_{v'}(c_v) \\
        &\geq \frac{1}{n} \sum_{v \in V} \sum_{v' \in V_{\ell_v}} v'^T c_v \\
        &\geq \frac{1}{n} \sum_{v \in V} \sum_{v' \in V_{\ell_v}} (v')_{\ell_v} (c_v)_{\ell_v}.
    \end{align*}
    From here we will rearrange the sum to range over the groups $v^\ell$. Observe that $v^\ell c_v^\ell \geq d^{-3/2}$ implies that $c_v^\ell \geq \frac{1}{v^\ell d^{3/2}}$. Then continuing,
    \begin{align*}
        &= \frac{1}{n} \sum_{\ell} \sum_{v,v' \in V_{\ell}} v^\ell c^{v'}_\ell \\
        &\geq \frac{1}{n} \sum_{\ell} \sum_{v,v' \in V_{\ell}} v^\ell \frac{1}{d^{3/2} v'_\ell} \\
        &\geq \frac{1}{n} \sum_{\ell} \sum_{v,v' \in V_{\ell}: v^\ell \geq v'_\ell} v^\ell \frac{1}{d^{3/2} v'_\ell} \\
        &\geq \frac{1}{n} \sum_{\ell} \sum_{v,v' \in V_{\ell}: v^\ell \geq v'_\ell} v^\ell \frac{1}{d^{3/2} v^\ell} \\
        &\geq \frac{1}{n \cdot d^{3/2}} \sum_\ell \frac{1}{2} |v^\ell|^2 \\
        &= \frac{n}{2 d^{3/2}} p^T p \\
        &\geq \frac{n}{2 d^{5/2}}.
    \end{align*}
    The claim follows.
\end{proof}

The insights about the uniform projection rule and the stable lottery rule can also be adapted to the the setting of $\elltwo$-normalized voter and candidate embeddings.

\begin{definition}[Uniform Projection Rule (\(\uprojtwo\)) for $\elltwo$ Normalization]
\label{def:uproj-ell2}
    Given candidate locations \( C \subset \mathbb{R}^d_{\ge 0} \), the \emph{uniform projection rule} \( \uprojtwo \) defines a distribution \( \{p_c\}_{c \in C} \) over candidates such that the expected pseudo-candidate vector
    $\hat{c} \defeq \sum_{c \in C} p_c \cdot c$
    minimizes the Kullback–Leibler (KL) divergence from the uniform vector, i.e.,
    \(
        \hat{c} = \arg\min_{x \in \ch(C)} \kl{\mu_2}{x}.
    \)
\end{definition}
Note that $\hat c$ in general will not satisfy $\| \hat c \|_2 = 1$, since the set of (even positive orthant) vectors with $\elltwo$-norm 1 is not convex.

\begin{theorem}\label{thm:middle-dist-ell2}
     The expected welfare of $\uprojtwo$ is at least $\sw(\uprojtwo) \geq \frac{n}{d}$.
     As a consequence, $\dist(\uprojtwo)=O(d)$.
\end{theorem}

\begin{proof}[Proof of \Cref{thm:middle-dist-ell2}]
First observe that
\[
\kl{\mu_2}{x}
=\sum_{i=1}^d\mu_2^i\ln\frac{\mu_2^i}{x^i}
={\sum_i\mu_2^i\ln\mu_2^i}
\;-\;\sum_{i=1}^d\mu_2^i\ln x^i,
\]
so minimizing $\kl{\mu_2}{x}$ over the convex set $\ch(C)$ is
equivalent to minimizing the smooth, convex function
\(
f(x)=-\sum_{i=1}^d\mu_2^i\ln x^i
\)
subject to $x\in \ch(C)$.  
The first‐order optimality condition gives 
\(
\nabla f(x^*)^\top(v-x^*) \;\ge\;0
\)
for every voter $v\in \ch(C)$.
Since
\(
\frac{\partial f}{\partial x^i}(x)
=-\frac{\mu_2^i}{x^i},
\)
we have
\[
\nabla f(x^*)^\top(v-x^*)
=-\sum_{i=1}^d\frac{\mu_2^i}{x^{*i}}(v^i - x^{*i})
\;\ge\;0 \]
\[
\Longrightarrow\;\;
\sum_{i=1}^d\frac{\mu_2^i\,v^i}{x^{*i}}\;\le\;\sum_{i=1}^d\mu_2^i.
\]
Since $\mu_2^i=1/\sqrt{d}$, this yields
\(
\sum_{i=1}^d\frac{v^i}{x^{*i}}\;\le\;d.
\)
Now define the auxiliary sequences
$a_i \defeq \sqrt{\frac{v^i}{x^{*i}}}$ and $b_i \defeq \sqrt{v^i\,x^{*i}}.$
Then
\(
\sum_{i=1}^d a_i\,b_i
=\sum_{i=1}^d v^i
\in [1,\sqrt{d}]
\) since $\|v\|_2=1$.
And, by Cauchy–Schwarz,
\begin{align*}
1 \le \Bigl(\sum_i a_i b_i\Bigr)^2
&\le\Bigl(\sum_i a_i^2\Bigr)\Bigl(\sum_i b_i^2\Bigr)\\
&=\Bigl(\sum_i\tfrac{v^i}{x^{*i}}\Bigr)\Bigl(\sum_i v^i\,x^{*i}\Bigr).
\end{align*}
Combining with $\sum_i v^i/x^{*i}\le d$ gives $\sum_{i=1}^d x^{*i}\,v^i \ge 1/d$.
Now as we have $n$ voters and the above inequality holds for arbitrary $v$, the total utilitarian welfare is at least $\frac{n}{d}$.

The bound on the distortion of $\uprojtwo$ then follows because the maximum utility for any voter is again at most $1$.
\end{proof}

As the original proof for linear stable lottery rule's distortion bound only relies on the bound of the uniform projection rule and the definition of stable lotteries, we immediately get the same result.  

\begin{definition}[Linear Stable Lottery Rule ($\lslrtwo$)]
\label{def:linear-slr-ell2}
    Given a stable lottery $\cW$ over committees of size $k=\sqrt{d}$, the \emph{linear stable lottery rule} $\lslrtwo$ on profile $\prof$ chooses each $c \in C$ with probability $\frac{1}{2\sqrt d}\Pr_{W \sim \cW(\prof)}[c \in W] + \frac{1}{2} \Pr_{c' \sim \uprojtwo(\prof)}[c' = c]$.
\end{definition}

\begin{theorem} 
\label{thm:stable-lottery-ell2}
    $\dist(\lslrtwo)=O(\sqrt{d}).$
\end{theorem}

Since \Cref{thm:middle-dist-ell2} implies that $\sw(c^*) \geq \frac{n}{d}$ for $\ell_2$ normalization, the proof of \Cref{thm:pure-stable-lottery-ell1} also proceeds identically.
We hence also have:
\begin{theorem} 
\label{thm:pure-stable-lottery-ell2}
    $\dist(\pslr)=O(d)$ for $\ell_2$-normalized candidates and voters.
\end{theorem}

\subsection{Optimizing distortion for $\elltwo$-normalization}

When switching to $\elltwo$ normalization, the key change is that the feasible region $\mathcal{F}$ becomes \textit{nonconvex}, due to the unit-norm constraint $\|v_j\|_2 = 1$ for each voter \( j = 1, \dots, n \). As a result, our instance-optimal problem likely cannot be solved exactly in polynomial time.

Nonetheless, several methods have been proposed for approximating such problems, including augmented Lagrangian techniques and projected gradient descent, which can be employed in practice to compute near-optimal solutions.

\section{Additional Experiments} \label{app:additional_experiment}

All experiments were run on a single‐socket Intel® Xeon® Gold 6130 CPU @ 2.10 GHz (16 physical cores, 32 threads), arranged in two NUMA nodes (cores 0–15 and 16–31).

In addition to evaluating the \textit{instance distortion} as in \Cref{sec:experiments}, here we also measure the \textit{empirical distortion}, defined as follows:
\begin{definition}[Empirical distortion]
    \label{def:empriical-distortion}
    Given a fixed utility profile $\uprof$, the \textit{empirical distortion} of a mechanism $f$ is the ratio of the optimal utilitarian social welfare to the expected welfare achieved by $f$ on the preference profile $\prof$ induced by $\uprof$. Formally,
    \[
        \dist(f, \uprof) \defeq  \frac{\sw(\uprof, c^*(\uprof))}{\mathbb{E}_{c \sim f(\prof(\uprof))} [\sw(\vec{u}, c)]}.
    \]
\end{definition}

While the distortion results presented in the main body establish the optimal guarantee of our proposed instance-optimal rules, the empirical findings offer further insight into the structure of real-world instances. A particularly notable observation is that even a simple rule such as plurality exhibits surprisingly strong performance in practice.

The general trends observed for empirical distortion closely mirror those of instance distortion across variations in $n$, $m$, and $d$: both measures remain largely stable. In terms of computational cost, the deterministic optimal rule exhibits the most significant increase in running time as parameters grow, followed by the randomized optimal rule. In contrast, the computation times for the remaining mechanisms are negligible in comparison. Notably, the number of alternatives $m$ emerges as the most influential factor affecting runtime.

\begin{figure}[ht]
    \centering
    \includegraphics[width=1\linewidth]{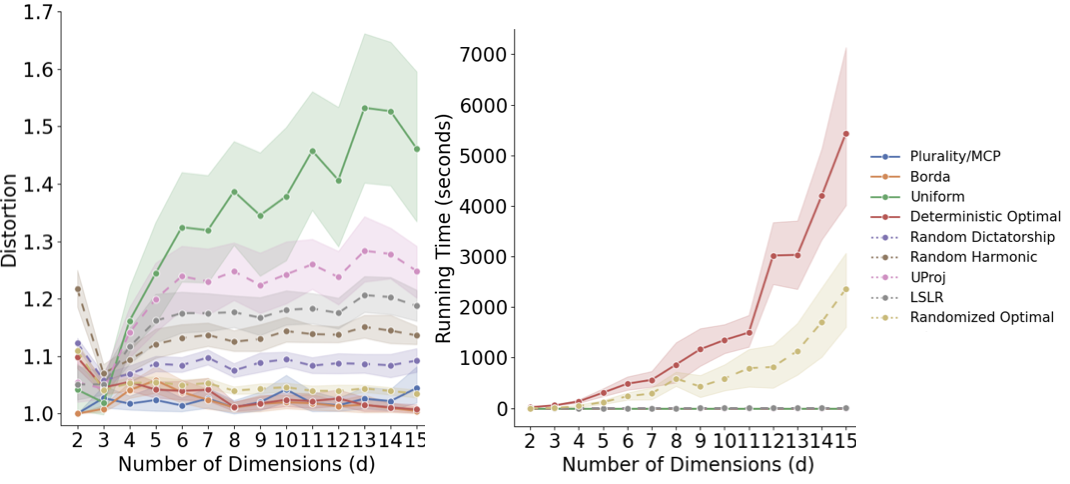}
    \caption{MovieLens Varying d (n=100,m=25). Left: Empirical Distortion; Right: Running Time.}
    \label{fig:mv_d_appendix}
\end{figure}

\begin{figure}[ht]
    \centering
    \includegraphics[width=1\linewidth]{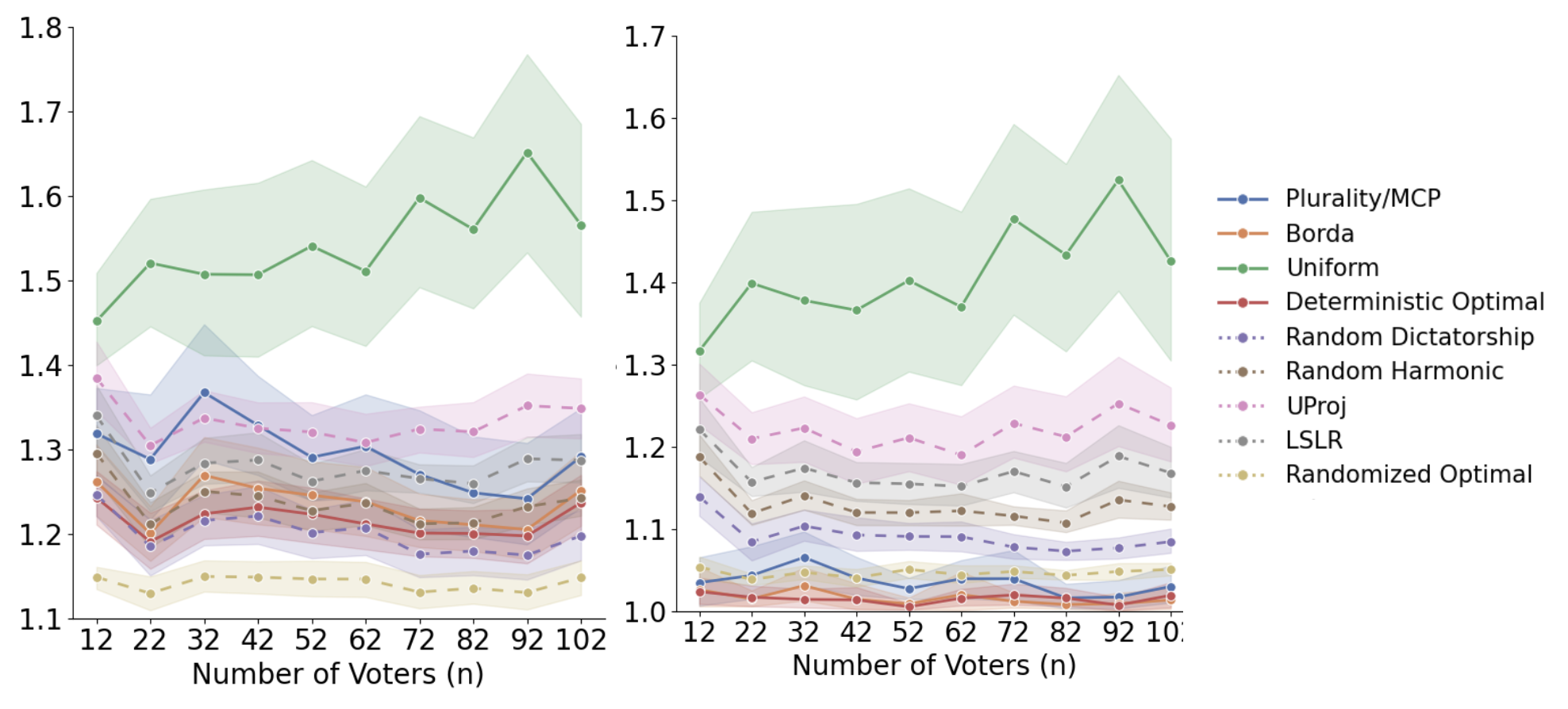}
    \caption{MovieLens Varying n (m=12, d=10). Left: Instance Distortion; Right: Empirical Distortion.}
    \label{fig:mv_n_dist}
\end{figure}

\begin{figure}[ht]
    \centering
    \includegraphics[width=1\linewidth]{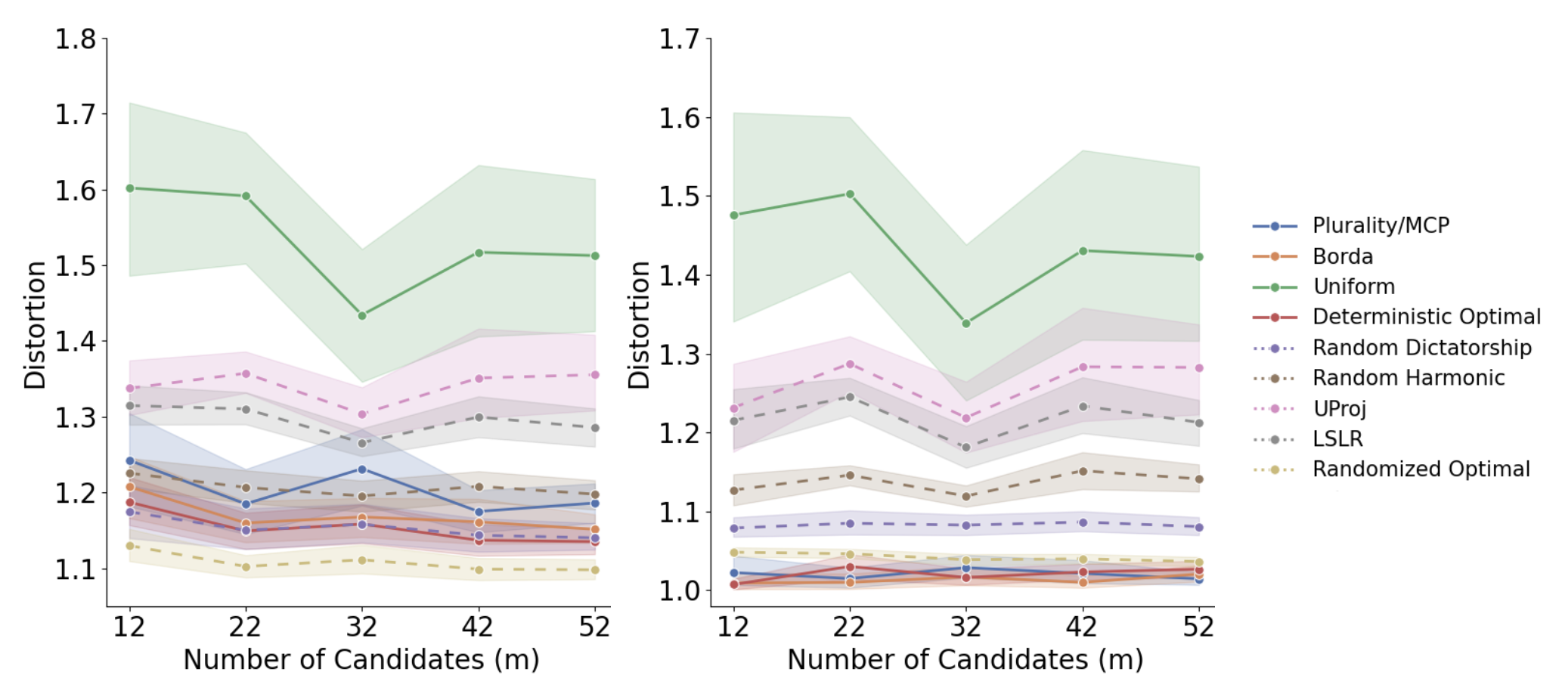}
    \caption{MovieLens Varying m (n=100,d=10). Left: Instance Distortion; Right: Empirical Distortion.}
    \label{fig:mv_m_dist}
\end{figure}

\begin{figure}[ht]
    \centering
\includegraphics[width=1\linewidth]{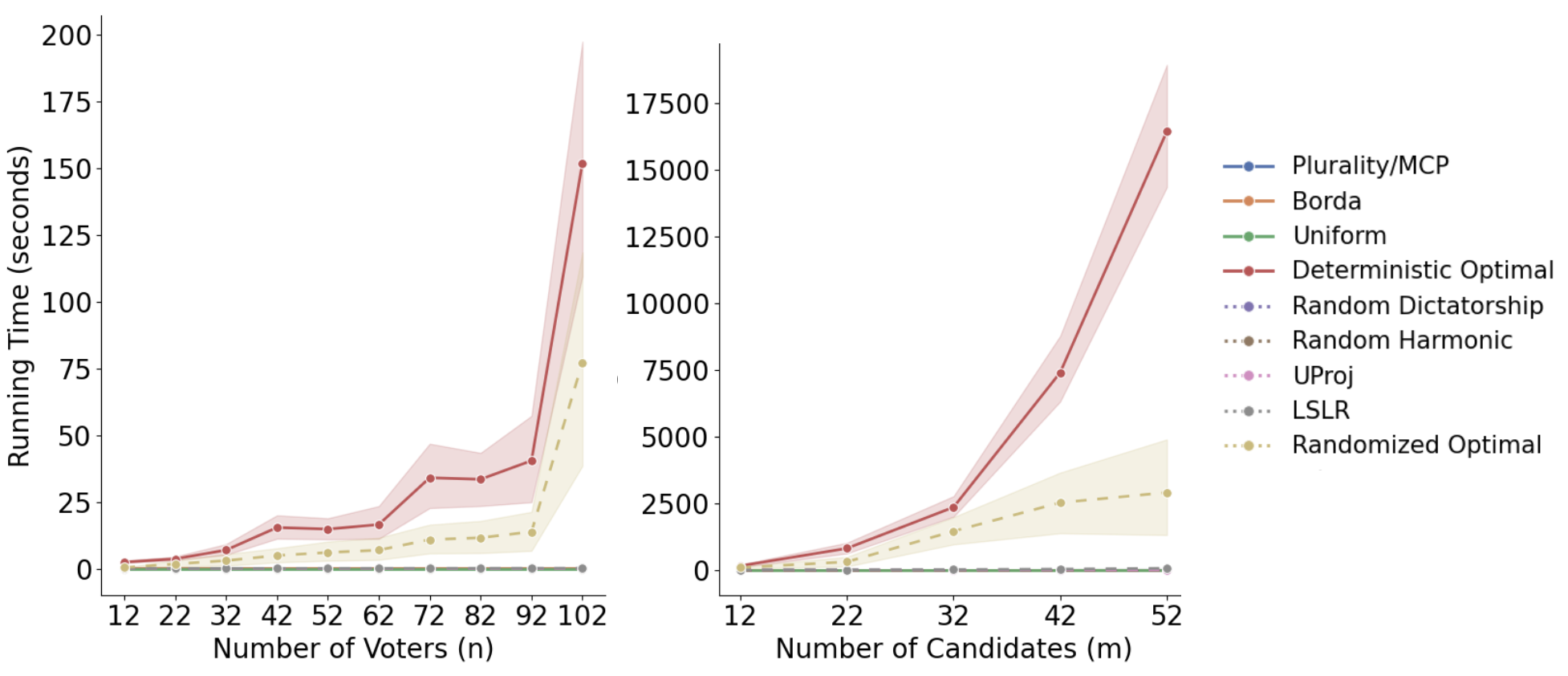}
    \caption{MovieLens Running Time. Left: Varying n (m=12, d=10); Right: Varying m (n=100, d=10).}
    \label{fig:mv_n_m_time}
\end{figure}

\begin{figure}[ht]
    \centering    \includegraphics[width=1\linewidth]{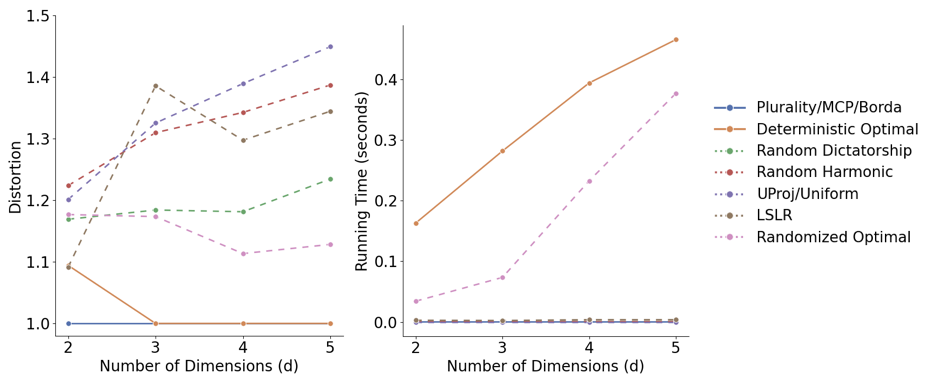}
    \caption{Abortion Opinion Survey. Left: Empirical Distortion; Right: Running Time.}
    \label{fig:abortion_rest}
\end{figure}

\end{document}